\documentclass[USenglish, a4paper, thm-restate,numberwithinsect, cleveref]{lipics-v2021}

\usepackage[utf8]{inputenc}
\usepackage[T1]{fontenc}
\usepackage[arrow, matrix, curve]{xy}

\usepackage{amsmath,amsthm,amssymb}

\usepackage{tikz}
\pgfdeclarelayer{background}
\pgfdeclarelayer{foreground}
\pgfsetlayers{background,main,foreground}

\usepackage[disable]{todonotes}
\usepackage{graphicx}
\usepackage{float}
\usepackage{dsfont}
\usepackage{braket}
\usepackage{enumitem}
\usepackage{tikz}
\usepackage{verbatim}
\usepackage{mathrsfs} 
\usepackage{xspace}
\usepackage{mathtools}
\usepackage{nicefrac}
\usepackage{lineno}

{\itshape}{\rmfamily}

\usepackage{cleveref}
\usepackage{apxproof}
\usepackage[noend]{algpseudocode}
\crefformat{brrule}{Branching Rule~#1}
\crefformat{redrule}{Reduction Rule~#1}
\crefformat{observation}{Observation~#1}
\crefformat{definition}{Definition~#1}
\crefformat{corollary}{Corollary~#1}
\crefformat{proposition}{Proposition~#1}
\crefformat{lemma}{Lemma~#1}
\crefformat{theorem}{Theorem~#1}

\usetikzlibrary{calc}

\newtheorem{brrule}{Branching Rule}[section]
\newcommand{\Oh}{\ensuremath{\mathcal{O}}}
\newcommand{\Wsf}{\ensuremath{\mathcal{W}}}
\newcommand{\mT}{\ensuremath{\mathcal{T}}}

\DeclareMathOperator{\swap}{swap}
\DeclareMathOperator{\branch}{branch}
\DeclareMathOperator{\val}{val}

\definecolor{myred}{rgb}{1,0.25,0.25}

\newcommand{\prob}[3]{\begin{quote}  \textsc{#1}\\  \textbf{Input:} #2\\  \textbf{Question:} #3\end{quote}}

\newcommand{\taskprob}[3]{\begin{quote}  \textsc{#1}\\  \textbf{Input:} #2\\  \textbf{Task:} #3\end{quote}}

\newcommand{\W}[1]{\ensuremath{\mathrm{W}[#1]}\xspace}
\newcommand\NP{\ensuremath{\mathrm{NP}}\xspace}

\newcommand\FPT{\ensuremath{\mathrm{FPT}}\xspace}

\newcommand{\CL}{\textsc{Clique}\xspace}

\newcommand{\WVC}{\textsc{Weighted Vertex Cover}\xspace}

\newcommand{\LWVClong}{\textsc{LS Weighted Vertex Cover}\xspace}
\newcommand{\GLWVClong}{\textsc{Gap LS Weighted Vertex Cover}\xspace}
\newcommand{\LVClong}{\textsc{LS Vertex Cover}\xspace}
\newcommand{\GLVClong}{\textsc{Gap LS Vertex Cover}\xspace}

\newcommand{\LWVC}{\textsc{LS WVC}\xspace}
\newcommand{\GLWVC}{\textsc{GLS WVC}\xspace}
\newcommand{\LVC}{\textsc{LS VC}\xspace}
\newcommand{\GLVC}{\textsc{GLS VC}\xspace}

\newcommand{\SWLVC}{\textsc{Knap-LS VC}\xspace}

\newcommand{\vc}{\ensuremath{\mathrm{vc}}}
\newcommand{\tw}{\ensuremath{\mathrm{tw}}}
\newcommand{\mw}{\ensuremath{\mathrm{mw}}}
\newcommand{\md}{\ensuremath{\mathrm{md}}}
\newcommand{\sw}{\ensuremath{\mathrm{sw}}}
\newcommand{\bor}{\ensuremath{\mathrm{border}}}

\newcommand{\mv}{\ensuremath{\mathcal{V}}}

\newcommand{\imp}{\ensuremath{\mathcal{\delta}}}
\newcommand{\leqBin}[2]{\ensuremath{\binom{#1}{\leq #2}}}
\newcommand{\goodswap}{good swap\xspace}

\nolinenumbers


\hideLIPIcs

\title{Parameterized Local Search for Vertex Cover: When only the Search Radius is Crucial}
\titlerunning{Parameterized Local Search for Vertex Cover}

\author{Christian Komusiewicz}{Friedrich Schiller University Jena, Institute of Computer Science, Germany}{c.komusiewicz@uni-jena.de}{https://orcid.org/0000-0003-0829-7032}{}

\author{Nils Morawietz}{Friedrich Schiller University Jena, Institute of Computer Science, Germany}{nils.morawietz@uni-jena.de}{https://orcid.org/0000-0002-7283-4982}{Supported by the Deutsche Forschungsgemeinschaft (DFG), project OPERAH, KO~3669/5-1.}

\authorrunning{Christian Komusiewicz, Nils Morawietz}
\Copyright{Christian Komusiewicz, Nils Morawietz}


\begin{document}
\maketitle
\begin{abstract}
A vertex set~$W$ in a graph~$G$ is a \emph{valid~$k$-swap} for a vertex cover~$S$ of~$G$ if~$W$ has size at most~$k$ and~$S'=(S \setminus W) \cup (W \setminus S)$, the symmetric difference of~$S$ and~$W$, is a vertex cover of~$G$.
If~$|S'| < |S|$, then~$W$ is~\emph{improving}.
In \LVClong, one is given a vertex cover~$S$ of a graph~$G$ and wants to know if there is a valid improving~$k$-swap for~$S$ in~$G$. In applications of \LVClong, $k$ is a very small parameter that can be set by a user to determine the trade-off between running time and solution quality. Consequently, $k$ can be considered to be a constant. Motivated by this and the fact that \LVClong is W[1]-hard with respect to~$k$, we aim for algorithms with running time~$\ell^{f(k)}\cdot n^{\Oh(1)}$ where~$\ell$ is a structural graph parameter upper-bounded by~$n$.  
We say that such a running time~\emph{grows mildly with respect to~$\ell$ and strongly with respect to~$k$}.
We obtain algorithms with such a running time for~$\ell$ being the~$h$-index of~$G$, the treewidth of~$G$, or the modular-width of~$G$. In addition, we consider a novel parameter, the maximum degree over all quotient graphs in a modular decomposition of~$G$.
Moreover, we adapt these algorithms to the more general problem where each vertex is assigned a weight and where we want to find a valid~$d$-improving~$k$-swap, that is, a valid~$k$-swap which decreases the weight of the vertex cover by at least~$d$.
\end{abstract}

\section{Introduction}
Local search is one of the most successful heuristic strategies to tackle hard
optimization problems~\cite{CSLS13,HS04,LHC+20}. Consequently, understanding when local search yields good results
and improving local search approaches is of utmost importance. In its easiest form, local
search follows a hill-climbing approach on the space of feasible solutions of the
optimization problem at hand. In this setting, one chooses some initial feasible solution
and then iteratively replaces the current solution by a better one in its local neighborhood until reaching a
local optimum, that is, a solution that has no better solution in its
neighborhood. 
Intuitively, it is clear that the larger the local search neighborhood, the
better the final solution will be. 
At the same time, searching
a larger neighborhood takes longer. In particular, for hard optimization problems, the running time will be superpolynomial when the neighborhood is too large. 
As a consequence, there is a trade-off between running
time and solution quality that is governed by the size of the local search neighborhood.

Parameterized local search offers a framework that may guide the design process for algorithms that attempt to search larger local neighborhoods. 
When applying parameterized local search to an optimization problem, the first step is to define a measure of distance between solutions. 
The local search neighborhood of a solution is then the set of solutions within distance at most~$k$. 
Here, $k$ is an \emph{operational} parameter that can be set by the user and that does not depend on the input data. 
The hope is now that the superpolynomial part of the running time for searching the local neighborhood depends solely on~$k$ while the dependence on the input size~$n$ is only polynomial. 
More precisely, the ultimate goal of parameterized local search is to devise an algorithm that determines
in~$f(k)\cdot n^{\Oh(1)}$~time whether there exists a better solution within distance at most~$k$ of the current one. 
Parameterized algorithms with this running time are called FPT-algorithms~\cite{DF13,C+15}.  Most local search problems, however, turn out to be \W1-hard with respect to the parameter~$k$~\cite{BIJK19,FFL+12,GHK14,GHNS13,Marx08,Szei11}. This makes it unlikely that they admit an FPT-algorithm when parameterized by~$k$~\cite{DF13,C+15}.

When applying local
search to \textsc{Vertex Cover}, the set of feasible solutions of a graph~$G=(V,E)$ is
naturally defined as the collection of vertex covers of~$G$, that is, vertex sets~$S\subseteq V$ that cover all edges of the graph.
The most obvious choice for a local
search neighborhood is the $k$-swap neighborhood. 
Here, two vertex sets~$S$
and~$S'$ are $k$-swap neighbors if and only if~$(S\setminus S')\cup (S'\setminus S)$ has
size at most~$k$. 
The problem of deciding whether a given vertex cover~$S$ of a graph~$G$
has a smaller vertex cover in its~$k$-swap neighborhood, called \LVClong, is \W1-hard with respect to~$k$~\cite{FFL+12}. Thus, at first it may seem unlikely that
parameterized local search can be successfully applied to \LVClong. 
There
are, however, some positive results for \LVClong. 
In particular, \LVClong admits an FPT-algorithm for $\Delta(G)+k$, where $\Delta(G)$
is the maximum degree of the input graph~\cite{FFL+12}.
That is, it can be solved in~$f(\Delta(G),k)\cdot n^{\Oh(1)}$~time. While this running time bound is certainly interesting for bounded-degree graphs, it
does not necessarily deliver on the promise of parameterized local search that the
superpolynomial part of the running time depends mostly on~$k$: 
for example~$f(\Delta(G),k)$ could be~$2^{\Delta(G)\cdot k}$. For many real-world instances, such a running time would be prohibitively large already for~$k=2$~\cite{KMSS25}. It is also known, however,
that \LVClong can be solved in time~$\Oh(2^k \cdot (\Delta(G) - 1)^{k/2}\cdot k^3 \cdot n)$~\cite{KK17}. 
In this running time only~$k$ appears in the exponent, while~$\Delta(G)$ appears solely in the
base of the exponential function. 
Consequently, for small values of~$k$ the running time guarantee
can still be practically relevant, even when~$\Delta$ is not too small. 
In particular, the latter
running time bound is polynomial when~$k$ is constant. The practical usefulness of the algorithm with this running time
was confirmed by experiments which showed that \LVClong can be solved
efficiently for~$k$ up to~25~\cite{KK17}.

\subparagraph{Our Setting.}
We aim to find further algorithms for \LVClong
that achieve running times like the one above which can be considered practical even though the superpolynomial running time part depends not only on the operational parameter~$k$ but also on some structural parameter~$\ell$. In particular, we want running times~$f(k,\ell)\cdot n^{\Oh(1)}$ where~$f$ is polynomial whenever~$k$ is constant. 
The class of functions~$f$ that provide this guarantee can be formalized as follows. 
\begin{definition}\label{def:smrt}
  Let~$f\colon \mathbb{N}\times \mathbb{N}\to \mathbb{N}$ be a function. We say that~$f$ \emph{grows mildly with respect to~$\ell$ and strongly with respect to}~$k$ if~$f(\ell,k)\in \Oh(\ell^{\,g(k)})$ for some computable function~$g$ depending only on~$k$.
\end{definition}
In the language of Definition~\ref{def:smrt}, we are interested in obtaining FPT-algorithms where~$f$ grows strongly only with
respect to~$k$ and mildly with respect to some other parameter~$\ell$. In our opinion, the usefulness of this setting is not limited to local search problems. 
Instead, it may be useful whenever
\begin{itemize}
\item two parameters~$k$ and~$\ell$ are studied, 
\item $k$ is known to be very small on relevant input instances, \item $k$ is known to be much smaller than~$\ell$ on these instances, and
\item the problem is W[1]-hard with respect to~$k$.
\end{itemize}
As a final remark, observe that Definition~\ref{def:smrt} is reminiscent of the definition of XP-algorithms which are the algorithms with running time~$n^{g(k)}$. Hence, for~$\ell<n$ our desired algorithms can be viewed as improved XP-algorithms for the parameter~$k$.   
\subparagraph{Our Results.} We provide FPT-algorithms for \LVClong
parameterized by~$k$ and several structural parameters of~$G$. Besides~$k$, we consider the treewidth of the input graph~$G$, denoted by~$\tw(G)$, the $h$-index of the input graph~$G$, denoted by~$h(G)$, the modular-width of~$G$, denoted by~$\mw(G)$, and a novel parameter, the maximum degree over all quotient graphs in a minimum-width modular decomposition, denoted by~$\Delta_{\md}(G)$.
In all our FPT-algorithms, the running time grows strongly with respect to~$k$ and only mildly with respect to the particular structural parameter. Moreover, for all these algorithms, the running time depends only linearly on the size of the input graph.

Many of our algorithms actually solve the more \GLWVClong problem where the
input graph is vertex-weighted, the cost of a vertex cover is the sum
of its vertex weights, and we search for a swap that improves the
current solution by at least~$d$ for some input value~$d$.  Local
search approaches for \textsc{Weighted Vertex Cover} have been studied
from a more practical perspective which motivates our study of
weighted variants of~\LVC. In addition, for weighted local search
problems, there may be exponentially long chains of local improvements
before one finds a local optimum~\cite{JPY88} even for swaps of
constant size~\cite{KM25}. Here, using a gap-variant of local search
could reduce the number of necessary steps by increasing the
improvement per step. We now discuss the results in detail; an overview of the results is given in Table~\ref{tab:results}.

\begin{table}[h!]
\centering
\scalebox{.85}{
\begin{tabular}{l c c c c c c}
\hline
 & $\tw(G)$ & $\Delta(G)$ & $h(G)$ & $\mw(G)$ & $\Delta_{\md}(G)$ & $\sw(G)$ \\
\hline
GLS WVC & $r(G)^{k+1}$  & $k! \cdot (\Delta-1)^{k}$ & $k! \cdot (h-1)^{k}$ & $\mw(G)^{k+1}$ &  & $\sw(G)^{k+1}$ \\
GLS VC  & $r(G)^{\frac{k+d}{2}+1}$ & $k!\cdot (\Delta-1)^{\frac{k+d}{2}}$ & $k! \cdot (h-1)^{k}$  & $\mw(G)^{\frac{k+d}{2}+1}$ &  $(\Delta_{\md}(k^2+k))^k$  & $\sw(G)^{\frac{k+d}{2}+1}$ \\
  LS WVC     & $r(G)^{k+1}$ & $2^k \cdot (\Delta-1)^{k}$ & $k! \cdot (h-1)^{k}$  & $\mw(G)^{k+1}$ &  & $\sw(G)^{k+1}$ \\
  LS VC      & $r(G)^{\frac{k+1}{2}+1}$ & $2^k \cdot (\Delta-1)^{k/2}$~\cite{KK17} & $k! \cdot (h-1)^{k}$ & $\mw(G)^{\frac{k+1}{2}+1}$ & $(\Delta_{\md}(k^2+k))^k$ & $\sw(G)^{\frac{k+d}{2}+1}$ \\
\hline
\end{tabular}
}
\caption{Running time bounds for the different FPT-algorithms; the linear dependence on~$n$ and factors polynomial in~$k$ are omitted. In the running times for~$\tw(G)$, we assume that a tree decomposition of width~$r$ is provided.}
\label{tab:results}
\end{table}

The~$h$-index of a graph~$G$ is the largest number~$h$ such that~$G$ has at least~$h$ vertices with degree at least~$h$~\cite{ES12}. For \GLWVClong, we obtain an algorithm with running time~$\Oh(k!\cdot (h(G)-1)^{k}\cdot n)$. 
This can be seen as an improvement over the FPT-algorithm
for~$\Delta(G)$ and~$k$~\cite{KK17} since~$h(G)$ is never larger than~$\Delta(G)$ and often much smaller~\cite{KMSS25}. In fact, in many
real-world instances the input graphs are scale-free, and on scale-free graphs~$h(G)$ is drastically smaller than~$\Delta(G)$. 
Even in such graphs, in order to speak of an
improvement, it is imperative that the running time of the FPT-algorithm grows mildly with respect to~$h(G)$ and strongly with respect to~$k$: a running
time of~$\Oh(2^{h(G)\cdot k}\cdot n)$ would be less desirable than the previous one
for~$\Delta(G)$ and~$k$ since the exponent would not be confined to the operational
parameter~$k$. 
The FPT-algorithm for~$h(G)$ and~$k$ follows a two-step approach. First, we consider
all possibilities of how an improving swap may interact with the high-degree
vertices. 
Then, for each such possibility, we find an improving swap in the remaining
low-degree part of the graph. 
Intuitively, this can be done by applying the previous
algorithm for bounded-degree graphs. 
There are, however, two main obstacles that we need
to overcome, in particular to obtain a good running time dependence on~$h(G)$
and~$k$. 
First, the partial swap on the high-degree vertices may actually produce a
handicap which is why it is not sufficient to find an improving swap on the low-degree
part. 
Instead, we must find a swap that decreases the vertex cover size by at least~$d$ for
some number~$d$. 
In other words, we need to solve the above-mentioned gap-version of
\LVClong.  Second, the swap may be disconnected on the graph induced by
the low-degree vertices. 
To identify the connected components of the swap, we branch on
the neighborhoods of small connected swaps that are optimal connected swaps for their
size, the important observation being that no such swap is used only if at least one neighbor of such a swap is swapped.

The FPT-algorithm for~$\tw(G)$ and~$k$ has running time~$\Oh((\tw(G)^{3k} + k^2)\cdot  n)$. 
It is based on dynamic programming on the tree
decomposition. 
The main observation is that for a bag of the tree decomposition it is
sufficient to consider all possibilities of how an improving swap interacts with the
bag. 
This reduces the running time for \LVClong to~$\Oh((\tw(G)^{3 \cdot\lceil \frac{k}{2}\rceil} + k^2)\cdot  n)$.
Hence, compared to the algorithm for \GLWVClong, we are able to consider swaps of double the size. 

We then consider parameters that are related to modular decompositions. 
These parameters measure a different structural aspect, the similarity of neighborhoods in the graph, than treewidth or the degree-related parameterizations. 
In particular, they can be very small in dense graphs.     
To the best of our knowledge, this is the first time that such parameters are considered in the context of parameterized local search. 
For \GLWVClong we develop an FPT-algorithm with running time~$\Oh(\mw(G)^k \cdot k \cdot (\mw(G) + k) \cdot n + m)$, where~$\mw(G)$ is the modular-width of~$G$, the size of the largest vertex set of any quotient graph of a modular decomposition of~$G$. 
The algorithm is based on bottom-up dynamic programming over the decomposition and considers all possibilities how an improving swap may interact with a node of the decomposition. 
We then show an improvement of this algorithm in terms of the structural parameter. More precisely, we show that when processing a node of the decomposition in the dynamic programming algorithm, one may apply a branching algorithm to determine how the swap interacts with the current node. Superficially, this branching algorithm resembles the one for $\Delta(G)$ and~$k$ but including the information computed for other nodes of the decomposition requires to combine the branching with~\textsc{Knapsack} DP-algorithms. 
This gives an FPT-algorithm for the parameters~$\Delta_{\md}(G)$ and~$k$. 
Recall that~$\Delta_{\md}(G)$ is the maximum degree over all quotient graphs of a modular decomposition of minimum width which is upper-bounded by~$\mw(G)$. 
We believe that this novel parameter can be useful in further algorithmic applications of modular decompositions because it often takes on substantially smaller values than the maximum degree and the modular-width~\cite{KMSS25}. 
We remark that the presented algorithm for~$\Delta_{\md}(G)$ and~$k$ only solves the unweighted gap version of \LVC; an extension to~\GLWVClong~seems possible but somewhat tedious.
In a second improvement, we show that instead of modular-width one can also obtain FPT-algorithms when using the smaller splitwidth.  
We complement these algorithms by conditional lower bounds that are based on the assumption that matrix-multiplication-based algorithms for \textsc{Clique} are running-time-optimal~\cite{ABW18}. We show that under this assumption, we may not expect a very large improvement of the previously known and the new algorithms.

\section{Preliminaries}

For integers~$i$ and~$j$ with~$i \leq j$, we define~$[i,j] :=   \{k \in \mathbb{N}\mid i \leq k \leq j\}$.
For a set~$A$, we denote with~${{A}\choose {2}}:=   \{\{a,b\}\mid a \in A, b\in A, a \neq b\}$ the collection of all size-two subsets of~$A$.
For two sets~$A$ and~$B$, we denote with~$A \oplus B :=   (A \setminus B) \cup (B \setminus A)$ the \emph{symmetric difference} of~$A$ and~$B$.

\subparagraph{Graph Notation.}
An (undirected) graph~$G=(V,E)$ consists of a set of vertices~$V$ and a set of edges~$E \subseteq {{V}\choose {2}}$. For a given graph~$G$, we let~$V(G)$ denote its vertex set and~$E(G)$ its edge set.
For vertex sets~$S\subseteq V$ and~$T\subseteq V$, we denote with~$E_G(S,T) :=   \{\{s,t\}\in E \mid s\in S, t\in T\}$ the edges between~$S$ and~$T$ and we use~$E_G(S) :=   E_G(S,S)$ as a shorthand.
Moreover, we define~$G[S] :=   (S,E_G(S,S))$ as the~\emph{subgraph of~$G$ induced by~$S$}. 
For a vertex~$v\in V$, we denote with~$N_G(v):=   \{w\in V\mid \{v,w\}\in E\}$ the \emph{open neighborhood} of~$v$ in~$G$ and with~$N_G[v]:=   \{v\}\cup N_G(v)$ the \emph{closed neighborhood} of~$v$ in~$G$. 
Analogously, for a vertex set~$S\subseteq V$, we define~$N_G[S] :=   \bigcup_{v\in S} N_G[v]$ and~$N_G(S) :=   \bigcup_{v\in S} N_G(v)\setminus S$.
If~$G$ is clear from the context, we may omit the subscript.

\subparagraph{Vertex Cover Local Search.}
A vertex set~$S\subseteq V$ is a~\emph{vertex cover} of~$G$ if at least one endpoint of each edge in~$E$ is contained in~$S$.
Let~$S$ be a vertex cover of~$G$. 
Then, $V\setminus S$ is an~\emph{independent set} in~$G$, that is, the vertices of~$V\setminus S$ are pairwise non-adjacent.
A~\emph{$k$-swap}, for~$k\in \mathbb{N}$, is a vertex set~$W$ of size at most~$k$ and~$W$ is said to be \emph{valid for a vertex cover~$S$ in~$G$} if~$S \oplus W$ is also a vertex cover of~$G$.
For each valid swap~$W$ for~$S$ in~$G$, both~$W \cap S$ and~$W \setminus S$ are independent sets and~$N(W) \setminus S = N(W\cap S) \setminus S \subseteq W$.
A swap~$W$ is~\emph{connected} if~$G[W]$ is connected.
Let~$\omega\colon  V \to \mathbb{N}$ be a weight function.
For each vertex set~$X\subseteq V$, we set~$\omega(X) :=   \sum_{x\in X}\omega(x)$.
\begin{definition}
The~\emph{improvement} of~$W$ is defined as~$\imp^S_\omega(W) :=   \omega(W\cap S) - \omega(W \setminus S)$.
Moreover, $W$~is~\emph{improving} if~$\imp_\omega^S(W) > 0$ and~\emph{$d$-improving} for some~$d\in \mathbb{N}$ if~$\imp_\omega^S(W)\geq d$.
\end{definition}
If~$S$ or~$\omega$ are clear from the context, we may omit them.
Next, we formally define the local search problems for~\textsc{Vertex Cover} we consider in this work.

\prob{\LWVClong (\LWVC)}{A graph~$G=(V,E)$, a weight function~$\omega\colon  V \to \mathbb{N}$, a vertex cover~$S$ of~$G$, and~$k\in \mathbb{N}$.}{Is there a valid improving~$k$-swap~$W\subseteq V$ for~$S$ in~$G$?}

\prob{\GLWVClong (\GLWVC)}{A graph~$G=(V,E)$, a weight function~$\omega\colon  V \to \mathbb{N}$, a vertex cover~$S$ of~$G$, $k\in \mathbb{N}$, and~$d\in\mathbb{N}$.}{Is there a valid~$d$-improving~$k$-swap~$W\subseteq V$ for~$S$ in~$G$?}

Moreover, we define \GLVClong (\GLVC) as the special case of~\GLWVC where~$\omega(v) = 1$ for each~$v\in V$ and~$d\in[1,k]$, and~\LVClong (\LVC) as the special case of~\GLVC where~$d = 1$.
Note that for an instance of~\GLVC the improvement of a swap~$W$ is~$|W\cap S| -|W \setminus S|$.
Let~$I=(G,S,\omega,k,d)$ be an instance of~\GLWVC.
We say that~\emph{$W\subseteq V(G)$ is a \goodswap for~$I$}, if~$W$ is a valid~$d$-improving~$k$-swap for~$S$ in~$G$.
Further, we say that a \goodswap~$W$ for~$I$ is~\emph{minimal} if each proper subset~$W'$ of~$W$ is not a \goodswap for~$I$ and~\emph{minimum} if there is no \goodswap~$W'$ for~$I$ with~$|W'|<|W|$.

We may also consider the \emph{permissive} version of \LVC where one is not limited to search in the local neighborhood for a better solution. More precisely, in the permissive problem the algorithm may either output any better solution or report that the current solution is locally optimal.

\section{Basic Observations and Lower Bounds}
In this section, we first define the notion of swap-instances.
These are instances obtained from an instance~$I$ of~\GLWVC after applying some partial swap. 
Swap-instances will be useful for describing certain parts of our algorithms such as branching rules. 
We then make some observations on certain useful properties of improving swaps.
Those are mostly generalizations of known results for~\LVC.
Finally, we present our running time lower bounds for the considered parameters; these bounds also hold for the permissive version.

\subparagraph{Swap-Instances.}
In our algorithms, we may change instances by performing some partial swaps, for example during branching. 
We call the instance obtained by such an operation a~\emph{swap-instance}.
Intuitively, the swap-instance~$\swap(I,W)$ for an instance~$I$ of~\GLWVC and a (partial) swap~$W$ is the~\GLWVC-instance obtained as follows: First, swap~$W$. 
Now,~$W$ might not be a valid swap because some vertices of~$W$ may been moved to the independent set without moving all of their neighbors to the vertex cover. To repair the swap we swap these vertices, giving an extended swap set~$W'\supseteq W$. 
To simplify the instance, the set~$W'$ is then removed from the instance together with the neighbors of~$W\cap S$ in~$S$.
Finally, to maintain equivalence, the remaining budget~$k$ is decreased by the number of swapped vertices and the required improvement~$d$ is decreased by the improvement of~$W'$.

The unique inclusion-minimal superset~$W'$ of~$W$ that has to be swapped to again obtain a vertex cover is called the~\emph{extension of~$W$ with respect to~$I$}.
Note that the extension of~$W$ with respect to~$I$ is exactly~$W \cup (N(W) \setminus S)$, since each independent set vertex adjacent with at least one vertex of~$S$ that is swapped out of the vertex cover, has to be swapped to obtain a vertex cover.

Formally, swap-instances are defined as follows; an example of a swap-instance is shown in~\Cref{fig swap instance}.
\begin{definition}
  Let~$I=(G,\omega, S, k, d)$ be an instance of~\GLWVC and let~$W \subseteq V(G)$ be a~$k$-swap. Define~$W' :=   W \cup (N(W) \setminus S)$ to be \emph{the extension of~$W$ with respect to~$I$}.
  The~\emph{swap-instance for~$I$ and~$W$} is the instance $$\swap(I,W) :=   (G',\omega', S' :=  S \setminus W, k', d')$$ with~$G':=  G -(N(W \cap S) \cup W')$, $k' :=  k - |W'|$, $d':=  d - \imp(W') = d - \omega(W' \cap S) + \omega(W' \setminus S)$, and~$\omega'$ being the restriction of~$\omega$ to~$V(G')$.
\end{definition}

\begin{figure}[t]
\begin{center}
\begin{tikzpicture}[yscale=.6, xscale = 1.5]
\tikzstyle{knoten}=[circle,fill=white,draw=black,minimum size=7pt,inner sep=0pt]
\tikzstyle{blocked}=[circle,fill=black,draw=black,minimum size=7pt,inner sep=0pt]
\tikzstyle{bez}=[inner sep=0pt]

		\node[knoten] (v5)  [label=below:{$v_4$}] at (-1.75,1) {};
		\node[knoten] (v6)[label=below:{$v_5$}] at ($(v5) + (1,0)$) {};
		\node[knoten] (v2) [label=below:{$v_6$}] at ($(v6) + (1,0)$) {};
		\node[knoten] (v7)[label=below:{$v_7$}] at ($(v2) + (1,0)$)  {};
		\node[knoten] (v8) [label=below:{$v_8$}] at ($(v7) + (1,0)$)  {};
		\node[blocked] (v1) [label={$v_1$}] at ($(v2) + (0,3)$) {};	
		\node[blocked] (v3) [label={$v_2$}] at ($(v1) + (1,0)$) {};
		\node[blocked] (v4)[label={$v_3$}] at ($(v3) + (1,0)$) {};

		\draw[-, ultra thick] (v1) to (v3);
		\draw[-, ultra thick] (v1) to (v5);
		\draw[-, ultra thick] (v1) to (v6);
		\draw[-, ultra thick] (v2) to (v3);
		\draw[-, ultra thick] (v2) to (v4);
		\draw[-, ultra thick] (v3) to (v7);
		\draw[-, ultra thick] (v3) to (v8);
		\draw[-, ultra thick] (v4) to (v7);
		\draw[-, ultra thick] (v4) to (v8);

		\node (a)[label=left:{$W'$}] at ($(v5) + (.3,3)$) {};
		\node (b)[label=left:{$W$}] at ($.5*(v1) + .5*(v2) + (.3,0)$) {};

\begin{pgfonlayer}{background}

		\node (a)[label=left:{$G$}] at (0,6.5) {};
		\node (a2) [label=left:{$G'$}] at (4.25,6.5) {};
		\draw[rounded corners, fill=green!30] (-2.6, 5.8) --
  ++(4.2, 0) -- ++(0,-3) -- ++(-.75, 0)-- ++(0, -3.4) -- ++(-3.45,0) -- cycle;
  
\draw[rounded corners, fill=yellow!30] (-2.5, 5.5) rectangle (.75, -0.3) {};
\draw[rounded corners, fill=blue!30] ($(v1) + (-.4,1.2)$) rectangle ($(v2) + (.4,-1)$) {};

\end{pgfonlayer}

		\node (a3) at (2.8,2) {\huge $\leadsto$};
		\node[blocked] (w4)[label={$v_3$}] at (4.5,4) {};  
		\node[knoten] (w7)[label=below:{$v_7$}] at (3.5,1)  {};
		\node[knoten] (w8) [label=below:{$v_8$}] at (4.5,1)  {};
		\draw[-, ultra thick] (w4) to (w7);
		\draw[-, ultra thick] (w4) to (w8);    
		\end{tikzpicture}
\end{center}
\caption{An instance~$I:=  (G,S,k,d)$ of \GLVC (left) and the swap-instance~$\swap(I,W) :=   (G',S',k',d')$ (right)  obtained from the swap~$W:=  \{v_1, v_6\}$. The vertex cover vertices are black, the independent set vertices are white.
The green area contains the vertices of~$N(W \cap S) \cup W'$ which are in~$G$ but not in~$G'$.
Since~$W'$ has size 4 and contains only one vertex of~$S$, $k' :=   k - 4$ and~$d' :=   d + 2$.
Moreover, the vertex~$v_2$ is not contained in~$G'$ since~$v_1$ is adjacent to~$v_2$ and leaves the vertex cover, which implies that~$v_2$ cannot leave the vertex cover afterwards.
}
\label{fig swap instance}
\end{figure}
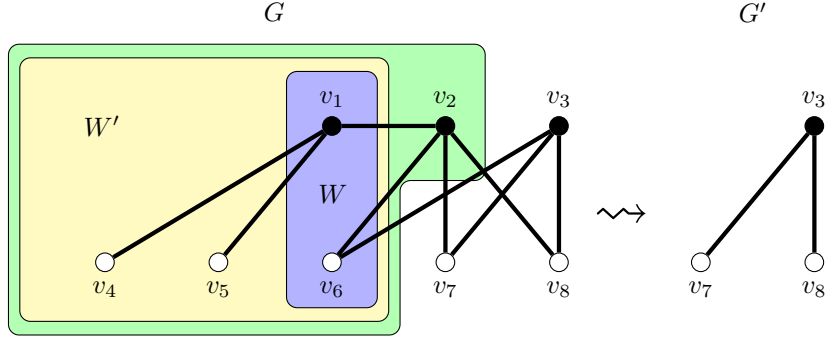
Informally, the following lemma states that the swap instance~$\swap(I,W)$ tells us whether there is any \goodswap containing some partial swap~$W$.

\begin{lemma}
Let~$I=(G=(V,E),\omega, S, k, d)$ be an instance of~\GLWVC and let~$W \subseteq V$ be a vertex set such that~$W\cap S$ is an independent set. 
There is a \goodswap~$W^*$ for~$I$ with~$W\subseteq W^*$ if and only if~$\swap(I,W)$ is a yes-instance of~\GLWVC.
\end{lemma}
\begin{proof}
Let~$(G',\omega', S', k', d') :=   \swap(I,W)$ and let~$W' = W \cup (N(W) \setminus S)$ be the extension of~$W$ with respect to~$I$.

Consider a \goodswap~$W^*$ for~$I$ with~$W\subseteq W^*$. 
Since~$W^*$ is valid, $W'\subseteq W^*$.
Moreover, no vertex of~$W^*\cap S$ has a neighbor in~$W\cap S = W' \cap S$.
Hence, $W^*\setminus W'$ contains only vertices of~$G'$, since~$G'$ contains all vertices of~$G$ except for the vertices of~$N(W\cap S)\cup W'$. 
Moreover, $W^*\setminus W'$ has size~$|W^*|-|W'| \leq k-|W'|  =  k'$, and~$\imp(W^*\setminus W') = \imp(W^*) - \imp(W') \geq d-\imp(W') = d'$.
As a consequence, $W^*\setminus W'$ is a \goodswap for~$\swap(I,W)$.

Let~$X$ be a \goodswap for~$\swap(I,W)$ and let~$W^* :=   X \cup W'$.
By definition, $W^*$ has size at most~$k$, $W^*$~is~$d$-improving, and~$W$ is a subset of~$W^*$.
Thus, it remains to show that~$W^*$  is a valid swap for~$S$ in~$G$.
Since i)~$W'$ is a valid swap for~$S$ in~$G$, ii)~$X$ is a valid swap for~$S'$ in~$G'$, and iii)~$G'$ contains all vertices of~$(V\setminus W')\setminus S$, we have the following: for each vertex~$v\in W^* \cap S$, each neighbor of~$v$ in~$V\setminus S$ is contained in~$W^*$.
Hence, $W^*$ is valid if there are no two adjacent vertices of~$S$ in~$W^*$.
Since~$W'\cap S$ and~$X\cap S$ are both independent sets in~$G$ and since no vertex of~$V(G')\cap S \supseteq X\cap S$ is adjacent to some vertex of~$W'\cap S$, this property holds.
Consequently, $W^*$ is valid. 
\end{proof}

Consider the swap-instance of an instance~$I:=  (G,S,k,d)$ of the unweighted problem \GLVC and some~$W\subseteq V(G)$.
Then~$k' + d' = k + d -  2 \cdot |W \cap S|$, since~$\imp(W') = -|W'| + 2 \cdot |W \cap S|$.
This observation has the following consequence for our branching algorithms: 
If in each branching step, we swap at least one vertex out of the vertex cover, the depth of the branching-tree is at most~$\frac{k+d}{2}$. 
Let~$I$ be an instance of~\GLWVC and let~$W$ be the subset of some valid swap.
When we replace the instance~$I$ by~$\swap(I,W)$ we may say that we~\emph{swap~$W$ in~$I$}.

\subparagraph{Properties of Improving Swaps.}
Next, we generalize some known properties of \goodswap{}s of~\LVC to the more general problems~\GLVC and~\LWVC.
Consider some improving swap~$W$ for~$S$ in~$G$.
Then, each connected component in~$G[W]$ is a valid swap and since~$W$ is improving, at least one connected component in~$G[W]$ is an improving swap for~$S$ in~$G$. 
Hence, the following holds.
\begin{observation}\label{connected}
Let~$I=(G,\omega,S,k)$ be an instance of~\LWVC.
If~$I$ is a yes-instance of~\LWVC, then there is a \goodswap~$W$ for~$I$ such that~$W$ is connected. 
\end{observation}

Next, we show that for the unweighted problem~\GLVC it is sufficient to consider instances where~$k+d$ is even.
\begin{lemma}\label{lem:parity}
Let~$I=(G,S,k,d)$ be an instance of~\GLVC where~$k+d$ is odd and~$k\geq 1$.
Then, $I$ is a yes-instance of~\GLVC if and only if~$I':=  (G,S,k-1,d)$ is a yes-instance of~\GLVC. 
\end{lemma}
\begin{proof}
$(\Rightarrow)$ 
Let~$W$ be a minimal \goodswap for~$I$.
Since~$d\geq 1$, $W$ contains a vertex~$v$ of~$S$.
Consider the swap~$W' :=   W\setminus \{v\}$.
The only difference between swapping~$W$ and~$W'$ is that~$W'$ keeps~$v$ in the vertex cover, so $W'$ is a valid~$k$-swap for~$S$ in~$G$.
Because~$W$ is minimal, $W'$ is not a \goodswap for~$I$.
Hence, $W'$ is not~$d$-improving.
This implies that~$\imp(W) = d$, since~$\imp(W) \geq d > \imp(W') = \imp(W) - 1$.
Next, we show that~$W$ has size at most~$k-1$ which then implies that~$W$ is a \goodswap for~$I'$.

Recall that the improvement of~$W$ is~$\imp(W) = |W\cap S| - |W\setminus S| = |W| - 2\cdot |W\setminus S|$.
Hence, $|W|$ is odd if and only if~$\imp(W)$ is odd, that is, $|W| + \imp(W)$ is even.
Recall that~$\imp(W)=d$ and that~$k+d$ is odd.
Consequently, $|W| + \imp(W) = |W| + d \leq k + d$ implies that~$W$ has size less than~$k$.
Hence, $W$ is a \goodswap for~$I'$.

$(\Leftarrow)$
Each~$(k-1)$-swap is a~$k$-swap.
Hence, each \goodswap for~$I'$ is a \goodswap for~$I$.
\end{proof}

Some of our algorithms branch over all possible intersections of a hypothetical \goodswap~$W$ with a given vertex set~$X$.
The following lemma shows that for~\GLVC, we only have to consider intersections of size at most~$\frac{k+d}{2}$ of~$X$ and~$W$.
Namely, we only have to consider the vertices~$S_X$ of~$W\cap X$ in~$S$ and the vertices~$C_X$ of~$W\cap X$ in~$V\setminus S$ that are not contained in the extension of~$S_X$. 
\begin{lemma}\label{lem: consider only small subsets}
Let~$I = (G,S,k,d)$ be an instance of~\GLVC and let~$W$ be a \goodswap for~$I$.
Further, for a set of vertices~$X\subseteq V(G)$, let~$S_X :=   W \cap X \cap S$ and~$C_X :=   W \cap X \setminus N[S_X]$. 
If~$|S_X \cup C_X| > \frac{k+d}{2}$, then~$W$ is not minimal.
\end{lemma}
\begin{proof}
First, we show that~$W^* :=   S_X \cup (N(S_X) \setminus S)$ is a \goodswap for~$I$.
Note that~$W^*$ contains no vertex of~$C_X$ and~$W^*$ is a (not necessarily proper) subset of~$W$, since~$W$ is a valid swap for~$S$ in~$G$ and contains all vertices of~$S_X$.
By definition, $W^*$ is a valid~$k$-swap for~$S$ in~$G$.
It remains to show that~$W^*$ is~$d$-improving.
To this end, note that each~$d$-improving~$k$-swap contains at most~$\frac{k-d}{2}$ vertices of~$V \setminus S$.
In particular, $|W\setminus S| = |C_X| + |N(S_X) \setminus S| \leq \frac{k-d}{2}$.
Since~$|S_X \cup C_X| > \frac{k+d}{2}$, this implies~$|S_X| > d + \frac{k-d}{2} - |C_X| \geq d + |(N(S_X) \setminus S)|$.
Hence, $W^*$ is~$d$-improving and thus a \goodswap for~$I$.

If~$W^*$ is a proper subset of~$W$, then the statement already holds.
Hence, assume in the following that~$W=W^*$.
This then implies that~$C_X = \emptyset$.
As a consequence, $S_X$ has size more than~$\frac{k+d}{2}$.
Note that this implies that the improvement of~$W$ is at least~$d+1$, since~$|W\setminus S_X| = |W| - |S_X| < k - \frac{k+d}{2} = \frac{k-d}{2}$.
Let~$v$ be an arbitrary vertex of~$S_X$.
Consider the swap~$W' :=   W\setminus \{v\}$.
Note that $W'$ is a valid~$k$-swap for~$S$ in~$G$ with~$\imp(W') = \imp(W) - 1 \geq d$.
Hence, $W'$ is a \goodswap for~$I$.
Moreover, since~$W'$ is a proper subset of~$W$, the statement holds.
\end{proof}

\subparagraph{Linear-time Solvable Special Cases.}
To obtain \FPT running times that are linear in the input size when the considered parameters are fixed, we handle instances with small values of~$k$ separately.

\begin{lemma}\label{lem: polytime param small}
\GLWVC can be solved in~$\Oh(n + m)$ time if~$k\leq 2$ and 
\GLVC can be solved in~$\Oh(n + m)$ time if~$k+d\leq 4$.
\end{lemma}

To show~\Cref{lem: polytime param small}, we show the two statements separately.
First, we show the statement for the weighted problem.

\begin{lemma}\label{lem: polytime k small}
\GLWVC can be solved in~$\Oh(n + m)$ time if~$k\leq 2$.
Moreover, for~$k\leq 2$, a valid~$k$-swap of maximum improvement can be found in $\Oh(n+m)$~time.
\end{lemma}
\begin{proof}
Let~$I=(G=(V,E),\omega, S, k,d)$ be an instance of~\GLWVC with~$k \leq 2$.
We first compute the set~$S^* :=   \{v\in S\mid N(v) \subseteq S\}$ of vertex cover vertices without independent set neighbors in $\Oh(n+m)$~time.
Moreover, if~$S^*\neq \emptyset$, then we also compute the vertex~$v^* \in S^*$ that maximizes~$\omega(v^*)$ in $\Oh(n)$~time.

If~$k=1$, then each \goodswap for~$I$ consists of a single vertex in~$S^*$.
Hence, $I$ is a no-instance of~\GLWVC if~$S^* = \emptyset$.
Otherwise, $\{v^*\}$ is a valid~$1$-swap of maximal improvement for~$I$.
Consequently, $I$ is a yes-instance of~\GLWVC if and only if~$\omega(v^*) \geq d$ and, thus, \GLWVC can be solved in $\Oh(n+m)$~time if~$k=1$.

If~$k=2$, then we additionally compute the set~$\mathcal{W}^* :=  \{\{v,w\} \mid v\in S, N(v) \setminus S = \{w\}, \omega(v) > \omega(w)\}$ in $\Oh(n+m)$~time.
Moreover, if~$\mathcal{W}^*\neq \emptyset$, then we also compute the swap~$W^* \in \mathcal{W}^*$ that maximizes~$\imp(W^*)$ in $\Oh(n)$~time.
Note that each minimal \goodswap either (i)~consists of a single vertex in~$S^*$, (ii)~consists of two non-adjacent vertices in~$S^*$, or (iii)~is a swap in~$\mathcal{W}^*$.
Hence, if~$\mathcal{W}^* \neq \emptyset$ and~$\imp(W^*) \geq d$, then~$I$ is a yes-instance of~\GLWVC.
Suppose that~$\mathcal{W}^* = \emptyset$ or~$\imp(W^*) < d$.
If~$S^* = \emptyset$, then~$I$ is a no-instance of~\GLWVC.
Thus, suppose that~$S^* \neq \emptyset$.
If~$\omega(v^*) \geq d$, then~$I$ is a yes-instance of~\GLWVC.
Otherwise, $I$ is a yes-instance of~\GLWVC if and only if there are two non-adjacent vertices of~$S^*$ of total weight at least~$d$.

In the following, we show that we can determine in $\Oh(n+m)$~time whether such a pair of vertices exists.
Moreover, if such a pair exists, then we can find one in the same running time.
Let~$S_{\geq} :=   \{v\in S^* \mid \omega(v) \geq d/2\}$ be the vertices of~$S^*$ of weight at least~$d/2$ and let~$S_{<} :=   \{v\in S^* \mid \omega(v) < d/2\}$ be the vertices of~$S^*$ of weight less than~$d/2$.
These sets can be computed in $\Oh(n)$~time.
Note that each \goodswap for~$I$ contains at least one vertex of~$S_{\geq}$.
If there are two non-adjacent vertices~$v$ and~$w$ in~$S_{\geq}$, then these vertices can be found in~$\Oh(n+m)$~time by checking for each vertex~$v$ of~$S_{\geq}$ if~$v$ has at most~$|S_{\geq}| - 2$ neighbors in~$S_{\geq}$.
In this case, since~$\omega(v) + \omega(w) \geq d$, $I$ is a yes-instance of~\GLWVC.
Otherwise, $S_{\geq}$ is a clique and~$\Oh(|S_{\geq}|^2)\subseteq \Oh(m)$.
Hence, we can sort the vertices of~$S_{\geq}$ according to their weight in $\Oh(m)$~time.
Since~$S_{\geq}$ is a clique, each \goodswap for~$I$ consists of one vertex~$v$ of~$S_{\geq}$ and one vertex~$w$ of~$S_{<}$ such that~$\{v,w\}\not\in E$.
To find such vertices, we check for each~$w\in S_{<}$, whether~$w$ is adjacent to each vertex of~$S_{\geq}$ and, if this is not the case, whether~$\omega(w) + \omega(v) \geq d$, where~$v$ is the vertex with the highest weight in~$S_{\geq} \setminus N(w)$.
If the latter is true for some~$v\in S_{<}$, $I$ is a yes-instance of~\GLWVC.
Otherwise, $I$ is a no-instance of~\GLWVC.
Since~$S_\geq$ is sorted, we only have to consider the first~$|N(w)| + 1$ vertices of highest weight in~$S_\geq$ to find the vertex~$v\in S_\geq \setminus N(w)$ of highest weight.
Hence, this last step can be performed in $\Oh(\sum_{w\in S_{<}} (|N(w)| + 1)) \subseteq \Oh(n+m)$~time and, thus, \GLWVC can be solved in $\Oh(n+m)$~time if~$k = 2$.
\end{proof}

Next, we show the statement for~\GLVC.

\begin{lemma}\label{lem: polytime k d small}
\GLVC can be solved in~$\Oh(n + m)$ time if~$k+d\leq 4$.
\end{lemma}
\begin{proof}
Let~$I=(G=(V,E), S, k,d)$ be an instance of~\GLVC with~$k  +d\leq 4$.
If~$k \leq 2$, then $I$ can be solved in $\Oh(n+m)$~time due to~\Cref{lem: polytime k small}.
Moreover, if~$d=0$, $I$ is a trivial yes-instance of~\GLVC.
Hence, in the following, we can assume that~$k>2$ and~$d>0$.
Since~$k+d\leq 4$, this then implies that~$k=3$ and~$d=1$.
Hence, $I$ is a yes-instance of~\GLVC if and only if (i) there is a vertex~$v$ of~$S$ with~$N(v) \subseteq S$ or (ii) there are two non-adjacent vertices~$w_1$ and~$w_2$ of~$S$ that have the same neighborhood in~$V\setminus S$ and this neighborhood consists of a single vertex.
In $\Oh(n+m)$~time, we can check if there is some vertex~$v$ of~$S$ with~$N(v) \subseteq S$.
If this is the case, answer yes.
Otherwise, we search for the vertices~$w_1$ and~$w_2$.
To this end, we remove all vertices of~$S$ from~$G$ that have at least two neighbors in~$V\setminus S$ and store for all remaining vertices of~$S$ the corresponding unique neighbor in~$V\setminus S$.
Since the vertices~$w_1$ and~$w_2$ we are looking for, have the same neighbor in~$V\setminus S$, we can also remove all edges between vertices of~$S$ that are adjacent to different vertices of~$V\setminus S$.
Hence, each connected component~$C$ in the resulting graph~$G'$ contains exactly one vertex~$v_C$ of~$V\setminus S$.
More precisely, $v_C$ is adjacent to all other vertices of~$C$.
If at least one of the connected components~$C$ is not a clique, then there are two non-adjacent vertices~$w_1$ and~$w_2$ of~$C$ such that~$N_{G}(w_1) \setminus S = N_{G'}(w_1) \setminus S = \{v_C\} =  N_{G'}(w_2) \setminus S = N_{G}(w_2) \setminus S$.
Hence, $\{w_1,v_C,w_2\}$ is a \goodswap for~$I$ and, thus, $I$ is a yes-instance of~\GLVC.
Otherwise, if~$C$ is a clique in~$G'$ for each connected component in~$G'$, then~$I$ is a no-instance of~\GLVC.
Note that all described steps can be performed by iterating over the edges a constant number of times.
Hence, this algorithm runs in $\Oh(n+m)$~time.
\end{proof}

\subparagraph{Lower Bounds.}
Let~$\omega$ be the matrix multiplication constant, that is, the smallest number~$\omega$ such that the product of~$n\times n$ matrices can be computed in time~$\Oh(n^{\omega})$. It is known that $2\le \omega < 2.373$~\cite{AW21}. 
Using a reduction to matrix multiplication, one can solve the \textsc{Clique} problem, which asks whether an~$n$-vertex graph has a clique of size~$k$, in $\Oh(n^{\omega\cdot k /3})$~time~\cite{NP85}. It is a long-standing question whether this running time can be improved to~$\Oh(n^{(\omega/3 - \varepsilon)k})$~\cite{ABW18,Woe04}.
Assuming that this is not the case, we obtain the following lower bounds for our considered problem. 
\begin{theorem}\label{thm:ClConjUnweight}
For every~$\varepsilon > 0$, every~$d\in[1,k]$, and~$\ell = \max\{n-\frac{k-d}{2}, \Delta(G), \vc(G), |S|, \mw(G)\}$, \GLVC cannot be solved in $\Oh(\ell^{(\omega/3 - \varepsilon)\cdot \frac{k+d}{2}})$ time, unless \textsc{Clique} can be solved in $\Oh(n^{(\omega/3-\varepsilon)\cdot k})$ time.
This holds also for the permissive versions of the problem.
\end{theorem}
\begin{proof}
Let~$\varepsilon > 0$ be a constant. 
We assume in the following that~$\varepsilon < \omega/3$, since the statement follows directly for~$\varepsilon \geq \omega/3$.
Moreover, let~$\widehat{I}=(\widehat{G}=(V,\widehat{E}), k)$ be an instance of~\CL with~$k \geq \frac{2}{(\omega/3)-\varepsilon}$ and let~$n$ denote the size of~$V$, and let~$d$ be an arbitrary value between~$1$ and~$k$.
We show that we can compute in $\Oh(n^2)$~time an equivalent instance~$I' = (G'=(V', E'), S, k',d)$ of~\GLVC such that~$\ell :=   \max\{|V'|-\frac{k'-d}{2},\Delta(G'), \vc(G'), |S|, \mw(G')\}$ is at most~$n$.
First, let~$G$ be the complement graph of~$\widehat{G}$, that is, $G:=  (V,E)$ with~$E :=   \binom{V}{2} \setminus \widehat{E}$.
Note that a set~$X\subseteq V$ is a clique in~$\widehat{G}$ if and only if~$X$ is an independent set in~$G$ and that one can compute~$G$ in $\Oh(n^2)$~time.
We can assume that the maximum degree of~$G$ is at most~$|V| - k$, since vertices degree at least~$|V| - k + 1$ are contained in no independent set of size~$k$.
We obtain~$G'$ by adding a set~$V^*$ of~$k - d$ new vertices to~$G$ such that~$N_{G'}(v) = V$ for each vertex~$v$ of~$V^*$.
Finally, we set~$k' :=   2k - d$ and~$S :=   V$, which completes the construction of~$I'$.
Note that this takes~$\Oh(n^2)$~time, since~$k\leq n$.
Next, we show that~$\widehat{I}$ is a yes-instance of~\CL if and only if~$I'$ is a yes-instance of~\GLVC.

$(\Rightarrow)$
Let~$C \subseteq V$ be an independent set of size~$k$ in~$G$, then~$S' :=   (V \setminus C) \cup V^*$ is a vertex cover of~$G'$ such that~$|S \oplus S'| = k'$ and~$|S'| \leq  |S| + |V^*| - |C| = |S| - d$.
Consequently, $I'$ is a yes-instance of~\GLVC.

$(\Leftarrow)$
Let~$W$ be a \goodswap for~$I'$ and let~$S':=  S\oplus W$.
Since~$W$ is a \goodswap for~$I$, $W$ has size at most~$k'=2k-d$ and~$|S'| < |S| -d$.
Consequently, $C :=   S \setminus S' = W\cap S$ is non-empty.
We show that~$C$ is an independent set of size~$k$ in~$G$. 
Since~$S'$ is a vertex cover of~$G'$ and every vertex of~$V^*$ is adjacent to every vertex of~$V$, it follows that~$W$ contains all vertices of~$V^*$.
By the fact that (i)~$W$ has size at most~$2k-d$, (ii)~$W$ contains all vertices of~$V^*$, and (iii)~$V^*$ has size~$k-d$, $W\cap S$ has size at most~$k$.
Moreover, since~$|S'| \leq |S| - d$, $C$ has size at least~$k' - |V^*| = k$.
As a consequence, $C$ has size exactly~$k$. 
Moreover, since~$S'$ is a vertex cover of~$G'$ and~$S'$ contains no vertex of~$C$, $C$ is an independent set in~$G$.
Consequently, $C$ is an independent set of size~$k$ in~$G$ and, thus, $\widehat{I}$ is a yes-instance of~\CL.

Next, we show that~$\ell :=   \max\{|V'| - \frac{k'-d}{2},\Delta(G'), \vc(G'), |S|, \mw(G')\}$ is at most~$n$.
By construction, $|V'|= n + k - d = n + \frac{k'-d}{2}$.
Since the maximum degree of~$G$ is at most~$n-k$, the maximum degree of~$G'$ is at most~$n$.
Moreover, since~$S$ is a vertex cover of size~$n$ of~$G'$, $\vc(G') \leq n$.
Next, we show that the modular-width of~$G'$ is at most~$n$.
Let~$(\mathcal{T}_1, \beta_1)$ be a modular decomposition of~$G$ and let~$(\mathcal{T}_2, \beta_2)$ be a modular decomposition of~$G'[V'\setminus V]$. 
Since each vertex of~$V$ is adjacent to each vertex of~$V'\setminus V = V^*$, a modular decomposition~$(\mathcal{T}, \beta)$ of~$G'$ can be obtained by combining~$(\mathcal{T}_1, \beta_1)$ and~$(\mathcal{T}_2, \beta_2)$ in the following way: 
We add a new root~$x^*$ where~$\beta(x^*)$ is a graph consisting of a single edge and the vertices of~$\beta(x^*)$ are the roots of the two modular decompositions~$(\mathcal{T}_1, \beta_1)$ and~$(\mathcal{T}_2, \beta_2)$.
Note that~$\mw(G) \leq n$. 
Moreover, since~$V'\setminus V$ is an independent set, we have~$\mw(G'[V'\setminus V]) = 2$.
Hence, the modular-width of~$G'$ is at most~$n$.

Now, if we have an algorithm~$A$ solving~\GLVC in~$\Oh(\ell^{(\omega/3 - \varepsilon)\cdot \frac{k+d}{2}})$ time, then \CL can be solved in $\Oh(n^{(\omega/3 - \varepsilon)k})$~time as well: 
Since~$k \geq \frac{2}{(\omega/3)-\varepsilon}$, the running time~$\Oh(n^{(\omega/3 - \varepsilon) \cdot k})$ dominates the time used to construct the instance~$I'$ of~\GLVC. 
Now the running time bound for solving \CL using~$A$ follows directly from~$\ell \leq n$ and~$\frac{k'+d}{2} = k$.
\end{proof}

For the extreme cases of \GLVC where~$d=1$ (this is precisely~\LVC) and~$d=k$, we may observe the following concrete running time bounds.
  \begin{corollary}\label{clique conjecture weighted}
    
    For every~$\varepsilon > 0$ and~$\ell:=  \max\{n-k+1, \Delta(G), \vc(G), |S|, \mw(G)\}$, \LVC cannot be solved in $\Oh(\ell^{(\omega/3 - \varepsilon)\cdot\lceil\frac{k}{2}\rceil})$ time  and \GLWVC cannot be solved in $\Oh(n^{(\omega/3 - \varepsilon)\cdot k})$~time, unless \textsc{Clique} can be solved in $\Oh(n^{(\omega/3-\varepsilon)\cdot k})$ time. This holds also for the permissive versions of the problem.

\end{corollary}

Based on the same reduction, we also derive the following by the fact that~\CL cannot be solved in $f(k) \cdot n^{o(k)}$~time for any computable function~$f$, unless the ETH fails~\cite{C+15}.

\begin{corollary}
\LVC cannot be solved in $f(k) \cdot n^{o(k)}$~time for any computable function~$f$, unless the ETH fails.
This holds even for the permissive version of~\LVC.
\end{corollary}
This running time lower bound of course also holds for the more general problems~\GLVC, \LWVC, and~\GLWVC.

\section{Parameterization by Treewidth}
In this section, we present FPT-algorithms for~$k$ and the treewidth of~$G$. 
As can be expected, the algorithms makes use of tree decompositions, for which we recall the definition in the following.

\subsection{Tree Decompositions}
A \emph{tree decomposition} of a graph~$G=(V,E)$ is a pair~$(\mT,\beta)$ consisting of a rooted tree~$\mT=(\mv, \mathcal{A}, x^*)$ with root~$x^*\in \mv$ and a function~$\beta\colon  \mv \to 2^V$ such that
\begin{enumerate}
\item for each vertex~$v$ of~$V$, there is at least one node~$x\in \mv$ with~$v\in \beta(x)$,
\item for each edge~$\{u,v\}$ of~$E$, there is at least one node~$x\in \mv$ such that~$\beta(x)$ contains~$u$ and~$v$, and
\item for each vertex~$v\in V$, the subgraph~$\mT[\mv_v]$ is connected, where~$\mv_v :=   \{x\in \mv\mid v\in \beta(x)\}$.
\end{enumerate}
We call~$\beta(x)$ the~\emph{bag} of~$x$.
The~\emph{width of a tree decomposition} is the size of the largest bag minus one and the \emph{treewidth} of a graph~$G$, denoted by~$\tw(G)$, is the minimal width of any tree decomposition of~$G$.

We consider tree decompositions with specific properties.
A node~$x\in \mv$ is called
\begin{enumerate}
\item a \emph{leaf node} if~$x$ has no child nodes in~$\mT$,
\item a \emph{forget node} if~$x$ has exactly one child node~$y$ in~$\mT$ and~$\beta(y) = \beta(x) \cup \{v\}$ for some~$v\in V\setminus \beta(x)$,
\item an \emph{introduce node} if~$x$ has exactly one child node~$y$ in~$\mT$ and~$\beta(y) = \beta(x) \setminus \{v\}$ for some~$v\in V\setminus \beta(y)$, or
\item a \emph{join node} if~$x$ has exactly two child nodes~$y$ and~$z$ in~$\mT$ and~$\beta(x) = \beta(y) = \beta(z)$.
\end{enumerate} 
A tree decomposition~$(\mT=(\mv, \mathcal{A}, x^*),\beta)$ is called~\emph{nice} if the bag of the root and the bags of all leaf nodes are empty sets and if every node~$x\in \mv$ is either a leaf node, a forget node, an introduce node, or a join node.
For a node~$x\in \mv$, we denote with~$V_x$ the union of all bags~$\beta(y)$, where~$y$ is contained in the subtree of~$\mT$ rooted at~$x$.
Moreover, we denote~$G_x :=   G[V_x]$ and~$E_x :=   E_G(V_x)$.

To obtain small polynomial factors in the running time, we first analyze the number of subsets of size at most~$k$ of any set~$X$ of size~$x$.
We denote by~$\leqBin{x}{k}$ the number of different subsets of~$X$ that have size at most~$k$, that is, $\leqBin{x}{k} :=   \sum_{r=0}^{\min(k,x)} \binom{x}{r}$.
\begin{lemma}\label{subsets of size k} 
Let~$k\geq 1$ be an integer and let~$X$ be an arbitrary set of size~$x\geq 3$.
Then, $\leqBin{x}{k} \leq 256 \cdot (x-1)^k/k^2$.
\end{lemma}
\begin{proof}
Note that~$\leqBin{x}{k} = \sum_{r = 0}^k \binom{x}{r} \leq 2^x$ and~$\leqBin{x}{k} \leq x^k$.
First, if~$k \leq 4$, then~$x^k \leq (2\cdot (x-1))^k \leq 16 \cdot (x-1)^k$ and since~$k^2 \leq 16$, $\leqBin{x}{k} \leq 256 \cdot (x-1)^k/k^2$. 
Second, if~$k \geq \max\{4, x/2\}$, then~$4^k \geq 2^x$ and~$2^k \geq k^2$.
Hence, if~$x \geq 9$, then by the fact that~$k\geq x/2$ and~$x\geq 3$, we get $$\leqBin{x}{k} \leq 2^x \leq 4^k \leq \frac{8^k}{2^k} \leq \frac{8^k}{k^2} \leq \frac{(x-1)^k}{k^2}.$$
If~$4 \leq x \leq 8$, then~$\leqBin{x}{k} \leq 2^x \leq 256 \cdot (x-1)^k/k^2$ since~$(x-1)^k > k^2$ for all~$k \geq 2$, and for~$x=3$, $\leqBin{x}{k} \leq 2^x = 8 \leq 256 \cdot 2^k/k^2$ for all~$k\geq 2$.
Finally, if~$4 < k < x/2$, then~$2 \cdot \binom{x}{k} \geq \sum_{r = 0}^k \binom{x}{r} = \leqBin{x}{k}$. 
Hence,
\begin{align*}
\leqBin{x}{k} &\leq 2 \cdot \binom{x}{k} \leq 2 \cdot \frac{x!}{(x-k)! \cdot  k!} \leq 2 \cdot x \cdot\frac{(x-1)!}{(x-k)! \cdot  k^2} \\&\leq 2 \cdot 2(x-1) \cdot   \frac{(x-1)^{k-1}}{k^2}  = 4\frac{(x-1)^{k}}{k^2} < 256 \cdot \frac{(x-1)^k}{k^2}.
\end{align*}
This completes the proof.
\end{proof}

\subsection{A Dynamic Programming Algorithm}
The algorithms are obtained by dynamic programming on a given tree decomposition of width~$r$, where each entry of the dynamic programming table considers the intersection of the current bag of size at most~$r+1$ with an improving swap~$W$ of size at most~$k$.
We deviate from this simple idea in one detail:
the algorithm  considers only (i)~the intersection~$W^S_x$ of~$W\cap S$ with the vertices of the current bag and (ii)~the intersection of~$W$ with those vertices of~$V\setminus S$ in the current bag that are not contained in~$N(W^S_x)$.
This is more technical but will give a running time improvement for the unweighted case.
\begin{theorem}\label{thm width and k fpt}
Let~$G=(V,E)$ be an undirected graph, let~$\omega\colon  V \to \mathbb{N}$ be a weight function, let~$S\subseteq V$ be a vertex cover of~$G$, let~$k$ be a natural number, and let~$(\mT=(\mv, \mathcal{A}, x^*), \beta)$ be a nice tree decomposition of width~$r$ for~$G$ with~$\Oh(r\cdot n)$ nodes.
One can compute in $\Oh((r^{k+1} + k^2)\cdot n)$~time a valid~$k$-swap~$W$ for~$S$ in~$G$ such that~$\imp(W)$ is maximal under all valid~$k$-swaps for~$S$ in~$G$. 
\end{theorem}

\begin{proof}
Due to~\Cref{lem: polytime k small}, the statement holds for~$k\leq 1$.
In the following, we show the running time by describing a dynamic programming algorithm for~$k \geq 2$.

Let~$N_x(U) :=   N(U) \cap \beta(x)$ denote the neighbors of~$U$ in the bag of~$x\in \mv$.
Recall that for a node~$x\in \mv$, the vertex set~$V_x$ is the union of all bags~$\beta(y)$, where~$y$ is a node of the subtree of~$\mT$ rooted at~$x$.
Moreover, recall that for a node~$x\in \mv$, $G_x :=   G[V_x]$ and~$E_x :=   E_G(V_x)$, where~$V_x$ is the union of all bags~$\beta(y)$, where~$y$ is contained in the subtree of~$\mT$ rooted at~$x$.

For each node~$x\in \mv$ in the tree decomposition, the dynamic programming table~$D_x$ has entries of type~$D_x[S_x, C_x,k']$ with $S_x \subseteq S \cap \beta(x)$, $C_x \subseteq \beta(x) \setminus (N(S_x) \cup S)$ and~$k' \in [0, k]$ such that~$|W_x| \leq k'$ where~$W_x :=   S_x \cup C_x \cup (N_x(S_x) \setminus S)$. 
 
Each entry stores the maximal improvement~$\imp_S(W)$ of a valid~$k'$-swap~$W \subseteq V_x$ for~$S\cap V_x$ in~$G_x$ such that~$W\cap S \cap \beta(x) = S_x$ and~$W \cap \beta(x) \setminus (N(S_x) \cup S) = C_x$.
 In other words, $W$ intersects with the vertices of~$S$ of the current bag exactly in~$S_x$ and~$W$ intersects with the vertices of~$V\setminus S$ of the current bag (minus the vertices that are contained in the extension of~$S_x$) exactly in~$C_x$.
Since we restrict~$W$ to be a~$k$-swap, the size of~$S_x\cup C_x$ is upper-bounded by~$k$ for each reasonable choice of~$S_x$ and~$C_x$.
Moreover, since both these subsets are disjoint, this implies that there are at most~$|\beta(x)|^k$ choices for~$S_x\cup C_x$ one has to consider in the dynamic programming table.

Since for $|S_x\cup C_x| > k$, $D_x[S_x,C_x,k']$ is not an entry of the dynamic programming table, we define a function~$f_x$ to prevent the evaluation of non-existing table entries.
For each nodes~$x\in \mv$, each subset~$S_x \subseteq \beta(x) \cap S$, each subset~$C_x \subseteq \beta(x) \setminus S$, and each~$k'\in[0, k]$,
we set~$f_x(S_x,C_x, k') :=   D_x[S_x, C_x,k']$ if~$|S_x \cup C_x| \leq k$ and~$f_x(S_x,C_x, k') :=  -\infty$, otherwise.

Next, we describe how to compute the entries of the dynamic programming tables.
For each leaf node~$\ell$ of~$\mT$, we fill the table~$D_\ell$ by setting~$D_\ell[\emptyset, \emptyset, k']:=   0$ for each~$k'\in[0,k]$.
This is correct, since~$G_\ell$ is the empty graph.

For all non-leaf nodes~$x$ of~$\mT$, we set~$D_x[S_x,C_x, k'] :=   -\infty$ if
\begin{itemize}
\item $S_x$ is not an independent set in~$G$, 
\item $|S_x \cup C_x \cup (N_x(S_x) \setminus S)| > k'$, or
\item $N(S_x) \cap C_x \neq \emptyset$.
\end{itemize}
Note that this is correct since in all three cases, there is no swap fulfilling the constraints of the table definition.
 To compute the remaining entries~$D_x[S_x,C_x,k']$, we distinguish the three types of non-leaf nodes.
 For each type, we omit the formal proof and only give an informal argument for of the correctness.
 
 \subparagraph{Forget Nodes.} Let~$x$ be a forget node, let~$y$ be the unique child of~$x$ in~$\mT$, and let~$v$ be the unique vertex in~$\beta(y) \setminus \beta(x)$. 
 The entries for~$x$ can be computed as follows:
 $$D_x[S_x, C_x, k'] :=   \begin{cases}
\max(f_y(S_x, C_x, k'), f_y(S_x \cup \{v\}, C_x \setminus  N(v), k')) & v\in S \text{~and} \\
\max(f_y(S_x, C_x, k'),f_y(S_x, C_x \cup \{v\}, k')) & v\notin S. 
 \end{cases}
 $$
Informally,  we chose the larger improvement of the best swap containing~$v$ and the best swap not containing~$v$.
To consider the best swap containing~$v$, we add~$v$ to the corresponding set ($S_x$ or~$C_x$) for the entry of the table~$D_x$.
If~$v$ is a vertex of~$S$, we also have to remove the vertices of~$N(v) \setminus S$ from~$C_x$, since these vertices are implicitly stored in the corresponding entry of~$D_y$ and, by definition, $N(S_y) \cap C_y = \emptyset$.

 \subparagraph{Introduce Nodes.} Let~$x$ be an introduce node, let~$y$ be the unique child of~$x$ in~$\mT$, and let~$v$ be the unique vertex in~$\beta(x) \setminus \beta(y)$. 
 For~$W_x :=   S_x \cup C_x \cup (N_x(S_x)\setminus S)$ and~$C^* :=   (N_x(v) \setminus S) \setminus N(S_x \setminus \{v\})$, the entries for~$x$ can be computed as follows:
 $$D_x[S_x, C_x, k'] :=   \begin{cases}
  f_y(S_x \setminus \{v\}, C_x \cup C^*, k' - 1) + \omega(v) & v\in W_x \cap S, \\ 
f_y(S_x, C_x \setminus \{v\}, k' -1) -\omega(v)  & v\in W_x\setminus S, \\ 
f_y(S_x, C_x, k')  & \text{otherwise,}  \end{cases}.
 $$

Informally, if~$v$ is a vertex of~$W_x$, we have to consider the entry of~$D_y$ where~$v$ is removed from the corresponding set ($S_x$ or~$C_x$) and adding the improvement we obtain from having~$v$ in the considered swap (increasing by~$\omega(v)$ if~$v\in S$ and decreasing by~$\omega(v)$ if~$v\notin S$).
If~$v$ is a vertex of~$S$, the vertices of~$C^*$ are not stored implicitly in~$D_y$, so we have to consider the entry of~$D_y$ where we also explicitly swap~$C^*$.
Otherwise, if~$v$ is not a vertex of~$W_x$, we consider the entry of~$D_y$ with the same subsets~$S_x$ and~$C_x$ and the same budget~$k'$.

 \subparagraph{Join Nodes.} Let~$x$ be a join node, let~$y$ and~$z$ be the unique children of~$x$ in~$\mT$.
 Recall that~$\beta(x) = \beta(y) = \beta(z)$. 
 For~$W_x :=   S_x \cup C_x \cup (N_x(S_x)\setminus S)$, the entries for~$x$ can be computed as follows:
 $$D_x[S_x, C_x, k'] :=  
\max_{0 \leq k'' \leq k' - |W_x|} D_y[S_x, C_x, k'' + |W_x|] + D_z[S_x, C_x , k'- k''] - \imp(W_x).
 $$
 
Informally, we consider all posibilities to distribute the budget~$k'$ among the two children:  
One part is assigned to the subset of vertices of~$W$ contained in the subtree rooted at~$y$ and one part is assigned to the subset of vertices of~$W$ contained in the subtree rooted at~$z$.
Note that~$W_x$ is contained in both of these vertex sets. This has two consequences: the sum of the budget for the two children is~$k'+|W_x|$ instead of~$k'$ and we have to remove~$\imp(W_x)$ from the obtained sum as it is added in the entries of both children.
 
The maximal improvement of any valid~$k$-swap for~$S$ in~$G$ can then be found in~$D_{x^*}[\emptyset, \emptyset, k]$.
Moreover, a corresponding swap~$W^*$ can be found via traceback:
$W^* \cap S$ consists of those vertices that are added to the set~$S_x$ in introduce nodes~$x$, and~$W^*\setminus S$ is exactly~$N(W^*\setminus S) \setminus S$.

It remains to show the running time.
Recall that~$(\mT=(\mv, \mathcal{A}, x^*), \beta)$ is a nice tree decomposition of width~$r$ for~$G$ with~$\Oh(r\cdot n)$ nodes. 
The number of entries of the table~$D_x$ is upper bounded by~$k+1$ times the number of subsets of~$\beta(x)$ of size at most~$k$.
Since for each node~$x\in \mv$, the bag $\beta(x)$ has size at most~$r+1$, all dynamic programming tables together contain~$\Oh(\leqBin{r+1}{k} \cdot k \cdot r \cdot n)$ entries.
Recall that~$\leqBin{r+1}{k}$ denotes the number of different subsets of size at most~$k$ of a set of size~$r+1$.
To complete the proof, it is sufficient to show that we can compute each of them in~$\Oh(k)$ time.
To this end, we perform the following preprocessing in which we assume that all considered subsets of~$V$ are stored as sorted lists.

First, we compute a degeneracy ordering~$\sigma$ of~$G$ in~$\Oh(n+m)\subseteq \Oh(n\cdot r)$~time.
Note that the degeneracy of~$G$ is never larger than the treewidth of~$G$.
Hence, with the help of~$\sigma$, we can check in~$\Oh(r)$ time whether two given vertices are adjacent.
Next, we compute the adjacency matrix~$A(x)$ of~$G[\beta(x)]$ for each node~$x \in \mv$.
Since~$\beta({x^*}) = \emptyset$ for the root vertex~$x^*$, we start with an empty adjacency matrix.
We now show that for each node~$x\in \mv$ where~$A(x)$ is already computed, we can compute~$A(y)$ for each child node~$y$ of~$x$ in~$\Oh(r^2)$ time.
Let~$x$ be a forget node with the unique child~$y$ and let~$v$ be the unique vertex of~$\beta(y) \setminus \beta(x)$.
Then, we can copy~$A(x)$ in~$\Oh(r^2)$ time and add a new row and a new column for the adjacency of~$v$.
To fill the new column and row, we only have to evaluate for each vertex~$w\in \beta(x)$, if~$v$ and~$w$ are adjacent.
These are at most~$r$ evaluations running in~$\Oh(r)$ time each.
Let~$x$ be an introduce node with the unique child~$y$ and let~$v$ be the unique vertex in~$\beta(x) \setminus \beta(y)$.
Then, we obtain~$A(y)$ by copying~$A(x)$ and removing the row and the column of~$v$ in~$\Oh(r^2)$ time. 
Let~$x$ be a join node with the unique children~$y$ and~$z$, then~$A(x) = A(y) = A(z)$  and we can obtain these copies of~$A(x)$ in $\Oh(r^2)$~time. 
Recall that~$\beta(x^*) = \emptyset$ for the root~$x^*$.
Hence, $A(x^*)$ is the empty matrix which can be computed in~$\Oh(1)$ time.
Since~$\mT$ contains~$\Oh(r\cdot n)$ nodes, we can compute~$A(x)$ for all nodes~$x \in \mv$ in~$\Oh(r^3\cdot n)$ time.
Since~$k> 2$, this can be upper-bounded by~$\Oh(r^k\cdot n)$~time.
Hence, in the following, we can check for each node~$x\in\mv$ in $\Oh(1)$~time whether two given vertices of~$\beta(x)$ are adjacent.
Further, for each node~$x\in \mv$, we perform the following steps:
\begin{itemize}
\item We store all independent sets~$S_x \subseteq \beta(x) \cap S$ of size at most~$k$.
 Using the adjacency matrix for~$\beta(x)$, we can check in~$\Oh(k^2)$ time whether a given set~$S_x$ of size at most~$k$ is an independent set. Hence, this part of the preprocessing runs in~$\Oh(\leqBin{r+1}{k} \cdot k^2)$ time.
\item For each independent set~$S_x \subseteq \beta(x) \cap S$ of size at most~$k$, we store the neighborhood~$N_x(S_x)\setminus S$ if it has size at most~$k$.
Otherwise, we store~$\bot$.
This preprocessing runs in $\Oh(\leqBin{r+1}{k} \cdot k^2)$~time:
First, consider all subsets~$S_x$ of size at most~$k-1$, this can be done in~$\Oh(r \cdot k)$ time since~$|\beta(x)| \leq r+1$ and~$|S_x| \leq k-1$.
Second, consider all subsets~$S_x$ of size exactly~$k$.
Choose an arbitrary vertex~$u\in S_x$.
Note that~$N_x(S_x) \setminus S = (N_x(S_x \setminus \{u\}) \setminus S) \cup (N_x(u) \setminus S)$.
Since these two subsets are already stored and have size at most~$k$ (or are set to~$\bot$) and the union of two sorted lists can be computed in linear time, for each such subset~$S_x$, this can be done in~$\Oh(k)$ time.
Hence, the total running time for this part of the preprocessing is $\Oh(\leqBin{r+1}{k}  \cdot k)$~time.

\item For each independent set~$S_x \subseteq \beta(x) \cap S$ of size at most~$k$ and each~$C_x \subseteq \beta(x) \setminus S$ with~$|S_x \cup C_x| \leq k$, we store the information whether~$N(S_x) \cap C_x = \emptyset$. 
Afterwards, we store the set~$W_x :=   S_x \cup C_x \cup (N_x(S_x) \setminus S)$ if it has size at most~$k$.
Otherwise, we store~$\bot$.
For each combination of sets~$S_x$ and~$C_x$, this preprocessing can be done in~$\Oh(k)$ time: 
Since~$|S_x \cup C_x| \leq k$, the set~$N_x(S_x) \setminus S$ is already stored. 
If~$|N_x(S_x) \setminus S| >k$, we immediately set~$W_x$ to~$\bot$.
Otherwise, $W_x$ is the union of three sets of size~$\Oh(k)$ which are given as sorted lists.
Since there are at most~$\leqBin{r+1}{k}$ combinations of sets~$S_x$ and~$C_x$, this part of the prepocessing runs in~$\Oh(\leqBin{r+1}{k}  \cdot k)$ time.

\item Let~$x$ be an introduce node with the unique child node~$y$ and the unique vertex~$v\in \beta(x) \setminus \beta(y)$.
Then, for each independent set~$S_x \subseteq \beta(x) \cap S$ of size at most~$k$ with~$|N_x(S_x)\setminus S| \leq k$, we store the set~$C^* :=   (N_x(v)\setminus S) \setminus N(S_x \setminus \{v\})$.
Now~$C^* = (N_x(v) \setminus S) \setminus N(S_x \setminus \{v\}) = (N_x(v) \setminus S) \setminus (N_x(S_x) \setminus S)$ and the two subsets of the latter expression are already stored and have size at most~$k$. 
Thus, $C^*$ can be computed in~$\Oh(k)$ time for each independent set~$S_x$ of size at most~$k$, since the difference of two sorted lists can be computed in linear time and~$|C^*| \leq k$.
Hence, this part of the preprocessing runs in~$\Oh(\leqBin{r+1}{k} \cdot k)$ time. 
\end{itemize}

Since there are~$\Oh(r\cdot n)$ nodes, the whole preprocessing (including constructing the degeneracy ordering and building all the individual adjacency matrices) runs in $\Oh((\leqBin{r+1}{k}\cdot k^2 + \leqBin{r+1}{k-1} \cdot r^2 \cdot k) \cdot n) = \Oh((\leqBin{r+1}{k}\cdot r\cdot  k^2 ) \cdot n)$~time.

Note that with this preprocessing, each entry of the tables~$D_x$ can be computed in~$\Oh(k)$ time:
For each forget or introduce node~$x$, one considers $O(1)$~cases, where for each case, one has to (i)~compute set operations for $\Oh(1)$~sets of size at most~$k$ each, as well as (ii)~evaluating one function call the function~$f_x(S_x,C_x,k')$.
Here, the latter can be done in $\Oh(k)$~time, since~$f_x(S_x,C_x,k')$ only checks whether~$|S_x\cup C_x| \leq k$.
For each join node~$x$, the corresponding set~$W_x$ can be computed in $\Oh(k)$~time and one computes the maximum of $\Oh(k)$~cases in which we combine two other table entries in $\Oh(1)$~time.
Moreover, note that all dynamic programming tables have $\Oh((\leqBin{r+1}{k}\cdot r \cdot k \cdot n)$~entries in total.
Since~$k > 2$ and due to~\Cref{subsets of size k}, we thus obtain that the whole algorithm runs in $\Oh(r^{k+1} \cdot n)$~time if~$r \geq 2$, and in $\Oh(2^r \cdot k^2\cdot n) = \Oh(k^2\cdot n)$~time, otherwise.
\end{proof}

Note that the running time of this algorithm can also be bounded by~$\Oh(2^r \cdot k^2 \cdot n)$ which implies that \GLWVC can be solved in polynomial time on graphs with a constant treewidth.
This will be in most applications irrelevant since one can find an optimal weighted vertex cover in the same running time~\cite{C+15,N06}; An exception would be the scenario where we explicitly want to find the best solution that is close to a given vertex cover. 

For the unweighted case, we may now exploit \Cref{lem: consider only small subsets} which tells us that it is sufficient to consider sets $S_x$ and~$C_x$ in the dynamic programming table that have total size at most~$\frac{k+d}{2}$.
This decreases the running time factor~$r^{k+1}$ to~$r^{\frac{k+d}{2}+1}$ for~\GLVC.
In particular, for the case of~$d=1$, that is, for~\LVC, this gives a substantial improvement of the exponential part of the running time from~$r^{k+1}$ to~$r^{\frac{k+1}{2}+1}$.

\begin{theorem}\label{thm: tw k d}
Let~$I=(G=(V,E),S,k,d)$ be an instance of~\GLVC.
When given a nice tree decomposition of width~$r$ for~$G$ with~$\Oh(r\cdot n)$ nodes, one can solve~$I$ in $\Oh((r^{\frac{k+d}{2} + 1} + k^2)  \cdot n)$~time. 
\end{theorem}
\begin{proof}
Due to~\Cref{lem: polytime k d small}, one can solve~$I$ in $\Oh(n+m)\subseteq \Oh(n\cdot r)$~time if~$k+d \leq 4$.
In the following, we may thus assume~$\frac{k+d}{2} > 2$.
To obtain the stated running time, we modify the dynamic program described in the proof of~\Cref{thm width and k fpt}.
We limit the entries of the dynamic programming table such that~$|S_x \cup C_x| \leq \frac{k+d}{2}$.
Hence, we also update~$f_x(S_x,C_x,k')$ to~$-\infty$ if~$|S_x \cup C_x| > \frac{k+d}{2}$.
To determine if there is a valid~$d$-improving~$k$-swap for~$S$ in~$G$ it is thus sufficient to check whether~$D_{x^*}[\emptyset, \emptyset, k]$ is at least~$d$, where~$x^*$ denotes the root of~$\mT$.
Moreover, the corresponding swap can be found via traceback.

The correctness of this modified dynamic program relies on~\Cref{lem: consider only small subsets} which, intuitively speaking, states that, if there is a \goodswap for~$I$, then such a swap can be found by only considering entries of the dynamic programming table where~$|S_x \cup C_x| \leq \frac{k+d}{2}$.

Since~$\frac{k+d}{2} \geq 2$, the adjacency matrix of~$G[\beta(x)]$ can be computed in~$\Oh(r^{\frac{k+d}{2}+1}\cdot n)$ time for all~$x \in \mv$.
Moreover, since we now only check for subsets of size at most~$\frac{k+d}{2}$, the preprocessing, filling all entries of the table, and checking~$D_{x^*}[\emptyset, \emptyset, k] \geq d$ can be done in~$\Oh(\leqBin{r+1}{\frac{k+d}{2}}\cdot r \cdot k^2 \cdot n)$ time which, due to~\Cref{subsets of size k}, is $\Oh(r^{\frac{k+d}{2}+1} \cdot n)$~time for~$r \geq 2$ and $\Oh(k^2 \cdot n + m ) \subseteq \Oh(k^2 \cdot n \cdot r)$~time, otherwise.
\end{proof}

Since computing a tree decomposition of minimal width is~\NP-hard, we cannot directly obtain a running time of~$\Oh((\tw(G)^{k+1} + k^2)\cdot n)$ and $\Oh((\tw(G)^{\frac{k+d}{2}+1} + k^2)\cdot n)$, respectively.
We can, however, compute a nice tree decomposition of width~$\tw(G)$ in~$\Oh(n + m)$ time if~$\tw(G) \leq 1800$~\cite{B96}.
Moreover, for each~$r\geq 0$, one can compute a nice tree decomposition of~$G$ of width at most~$1800\cdot r^2$ or correctly output that~$\tw(G) > r$ in~$\Oh(r^{7} \cdot n \cdot \log (n))$ time~\cite{FLSPW18}.
Hence, we can compute in $\Oh(\tw(G)^{8} \cdot n \cdot \log (n))$~time~\cite{FLSPW18} a nice tree decomposition of~$G$ of width at most~$1800\cdot \tw(G)^2$ by applying  the approximation algorithm for each~$r\in [1,\tw(G)]$.
If~$\tw(G) \geq 1800$, then the width of the latter tree decomposition is smaller than~$\tw(G)^3$.
Altogether, we obtain the following by~\Cref{lem: polytime k d small} and~\Cref{lem: polytime k small}.
\begin{corollary}
\GLVC can be solved in $\Oh((\tw(G)^{\frac{3\cdot(k+d)}{2}+1} + k^2) \cdot n\cdot \log(n))$~time and~\GLWVC can be solved in $\Oh((\tw(G)^{3k+1} + k^2)\cdot  n \cdot \log(n))$~time.
\end{corollary}
Note that even if a tree decomposition of width~$\tw(G)$ is given, one cannot improve much on the running time due to the lower bound of~\Cref{thm:ClConjUnweight}, since the treewidth of a graph is never larger than its vertex cover number. 
Due to this relation, we further obtain the following by initially computing a 2-approximation of a minimum vertex cover.

\begin{corollary}
For a given vertex cover~$S^*$ of~$G$, \GLVC can be solved in $\Oh(|S^*|^{\frac{(k+d)}{2}+1} \cdot n + m)$~time and~\GLWVC can be solved in $\Oh(|S^*|^{k+1}  \cdot  n + m)$~time.
In general, \GLVC can be solved in $\Oh((2\cdot\vc(G))^{\frac{(k+d)}{2}} \cdot n + m)$~time and~\GLWVC can be solved in $\Oh((2\cdot\vc(G))^k \cdot  n + m)$~time. 
\end{corollary}

Since~$S$ is a vertex cover of~$G$, this also implies an algorithm for~\GLVC with running time~$\Oh(|S|^{\frac{(k+d)}{2}+1} \cdot  n + m)$ and an algorithm for~\GLWVC with running time~$\Oh(|S|^{k+1} \cdot  n + m)$.

\section{Degree-Related Parameterizations}\label{sec:degree}
In this section, we present FPT-algorithms with running times that grow strongly with respect to~$k$ and only mildly with respect to~$\Delta(G)$ or the~$h$-index of~$G$. 
In contrast to previous work, these FPT-algorithms solve the more general problems with weights and gap-improvements; the algorithms for~$\Delta(G)$ will be used as subroutines in the algorithm for the~$h$-index of~$G$.
We start by presenting an algorithm for instances with an~$h$-index of at most~$1$ which will be used to handle border cases for both parameterizations.
\begin{lemma}\label{lem: polytime h small}
\GLWVC can be solved in~$\Oh(k\cdot \log(k) + n)$ time if the~$h$-index of~$G$ is at most~$1$.
\end{lemma}
\begin{proof}
Let~$I=(G=(V,E),\omega, S, k,d)$ be an instance of~\GLWVC with~$h(G) \leq 1$.

First, we present an algorithm with the stated running time for the case~$\Delta(G) \leq 1$. In that case,~$G$ consists of isolated vertices and isolated edges.
The first step is to apply the following initial reduction rule.
For each edge~$e:=  \{u,v\}\in E$ where~$S$ contains both~$u$ and~$v$, remove~$v$ from~$G$ if~$\omega(v) \leq \omega(u)$ and remove~$u$ from~$G$, otherwise.
Since each vertex cover of~$G$ contains at least one endpoint of~$e$, it is never optimal to swap the endpoint of~$\{u,v\}$ of smaller weight out of~$S$ and thus, the reduction rule produces an equivalent instance.
Moreover, this reduction rule can be applied exhaustively in $\Oh(n)$~time. 

Hence, in the following, we can assume that~$S$ contains exactly one endpoint of each edge of~$E$.
Since~$\Delta(G) \leq 1$, each valid connected swap consists either of (i)~a degree-0 vertex in~$S$ or (ii)~both endpoints of an edge.
 
We now compute~$\mathcal{W}_1 :=   \{\{v\}\mid v\in S, N(v) = \emptyset\}$ the set of valid~$1$-swaps and~$\mathcal{W}_2 :=   \{\{v,w\} \mid \{v,w\}\in E, \{v,w\}\not\subseteq S, \imp(\{v,w\}) > 0\}$ the set of valid connected improving~$2$-swaps for~$S$ in~$G$ in $\Oh(n)$~time.
Note that each vertex is contained in at most one set in~$\mathcal{W}_1 \cup \mathcal{W}_2$.
Since~$\Delta(G) \leq 1$, the union~$W$ of each subset of swaps of~$\mathcal{W}_1 \cup \mathcal{W}_2$ is a valid swap for~$S$ in~$G$.
Hence, it remains to find some set~$W$ of size at most~$k$ with the maximal improvement.
To this end, we compute the set~$\mathcal{W}_i^k$ of the~$k$ elements of~$\mathcal{W}_i$ of maximal improvement and sort them by weight for each~$i\in \{1,2\}$. 
This can be done by finding the swap~$W_i^k$ in~$\mathcal{W}_i$ with the~$k$th largest improvement using the median of the medians in $\Oh(n)$~time and afterwards filtering all swaps of improvement at least~$\imp(W_i^k)$ and sorting them in $\Oh(k\cdot \log(k) + n)$~time.
We now compute~$W$ greedily as follows. Start with an empty set~$W$. While~$k \geq 2$, let~$W_1$ and~$W_1'$ be the two swaps in~$\mathcal{W}_1^k$ with the maximal improvement and let~$W_2$ be the swap in~$\mathcal{W}_2^k$ with the maximal improvement.
If~$\imp(W_1 \cup W_1') \geq \imp(W_2)$, then update~$W$ to~$W \cup W_1\cup W_1'$ and remove~$W_1$ and~$W_1'$ from~$\mathcal{W}_1^k$ .
Otherwise, update to~$W \cup W_2$ and remove~$W_2$ from~$\mathcal{W}_2^k$.
In both cases decrease~$k$ by 2.
If~$k = 1$, then update~$W$ to~$W \cup W_1$ and~$k$ to~$0$, where~$W_1$ is the swap in~$\mathcal{W}_1^k$ with the maximal improvement.
If~$k = 0$, then we may add no further vertices to~$W$ and~$I$ is a yes-instance of~\GLWVC if and only if~$\imp(W) \geq d$.
Altogether this algorithm runs in~$\Oh(k\cdot \log(k) + n)$~time.

Next, we present an algorithm for the case~$\Delta(G) > 1$.
Since~$h(G) \leq 1$, there is exactly one vertex~$v^*$ of degree at least 2 in~$G$. We branch into the two possible cases whether~$v^*$ is contained in the swap~$W$ or not.
For the case that $v^*$ is in~$W$, we compute the swap-instance~$I_1:=  \swap(I, \{v_x\})$ 
and check whether~$I_1$ is a yes-instance of~\GLWVC.
For the case that $v^*$ is not in~$W$, we check whether the instance~$I_2$ obtained by removing~$v_x$ and, if~$v_x\in V \setminus S$, all vertices in~$N(v_x) \cap S$ from~$G$, is a yes-instance.
The instances~$I_1$ and~$I_2$ have a maximum degree of at most one, so we can solve both of them in $\Oh(k\cdot \log(k) + n)$~time using the algorithm for the case~$\Delta(G)\le 1$.
\end{proof}

Hence, for the following algorithms, we assume that~$h(G)$ and~$\Delta(G)$ are at least~$2$.

\subsection{Parameterizing Unweighted Gap Local Search by Maximum Degree}

In this section, we present an algorithm for~\GLVC whose running time grows mildly with respect to~$\Delta(G)$ and strongly with respect to~$k+d$.
Recall that~\LVC (the special case of~\GLVC with~$d=1$) can be solved in $2^k \cdot (\Delta(G) -1)^{\frac{k+1}{2}}\cdot n^{\Oh(1)}$~time.
Our running time bound for the general algorithm almost meets this bound for~$d=1$, the only difference being that the leading factor of~$2^k$ is replaced by~$k!$. 
\begin{theorem}\label{glvc delta k d}
\GLVC can be solved in $\Oh(k!\cdot (\Delta-1)^{\frac{k+d}{2}}\cdot n)$~time.
\end{theorem}
The first idea for an algorithm is obviously to adapt the known~$\Oh(2^k \cdot (\Delta -1)^{\frac{k+1}{2}}\cdot k^2 \cdot n)$-time algorithm for~\LVC~\cite{KK17} to~\GLVC. 
This algorithm, however, relies on the fact that for~\LVC, it is sufficient to consider only connected swaps.
For~$d$-improving swaps this is not the case: for~$k = 2$, an improvement of at least 2 can only be achieved by swapping two vertices out of the current solution that are not adjacent. 
These vertices may have an arbitrarily large distance in the graph. 
This makes the gap version of the problem considerably more complicated.

To avoid considering all possible vertex sets of size at most~$k$, we present two branching rules.
The first rule applies if there is a vertex~$v$ in~$S$ where~$N(v) \subseteq S$ and branches in all possible ways to swap either~$v$ or two non-adjacent vertices of~$N(v)$.
If this rule cannot be applied, then each vertex in~$S$ has at least one neighbor in~$V\setminus S$ and, thus, there is no valid improving swap of size one.

\begin{proposition}\label{correctness rule isolated}
Let~$I=(G,S,k,d)$ be a yes-instance of~\GLVC and let~$v$ be a vertex of~$S$ with~$N(v) \subseteq S$.
There is a \goodswap~$W$ for~$I$, such that (i)~$v$ is contained in~$W$ or (ii)~$W$ contains at least two neighbors of~$v$. 
\end{proposition}
\begin{proof}
Let~$W$ be a \goodswap for~$I$.
Suppose that~$v$ is not contained in~$W$ and that~$W$ contains at most one neighbor of~$v$, as otherwise the statement already holds.
If~$W$ contains no neighbor of~$v$, let~$w$ be an arbitrary vertex of~$S\cap W$.
Otherwise, let~$w$ be the unique vertex in~$W\cap N(v)$.
Note that~$w\in S \cap W$ by the fact that~$N(v) \subseteq S$.
Hence, $W' :=   W \setminus \{w\}$ is a valid~$(d-1)$-improving~$(k-1)$-swap for~$I$. 
Observe that~$W'$ contains neither~$v$ nor a neighbor of~$v$.
Consequently, $W^* :=   W' \cup \{v\}$ is a valid~$d$-improving~$k$-swap for~$S$ in~$G$.
\end{proof}

From~\Cref{correctness rule isolated}, we derive the following branching rule.
Here, we swap~$v$ or two independent neighbors of~$v$ in each case.
This gives a better branching vector than swapping only a single vertex of~$N[v]$ in each case.

\begin{brrule}\label{rule:deg one}
Let~$I=(G=(V,E),S,k,d)$ be an instance of~\GLVC and let~$v$ be a vertex of~$S$ with~$N(v) \subseteq S$.
For each swap~$W \in (\binom{N(v)}{2} \setminus E) \cup \{\{v\}\}$, branch into the case of swapping~$W$.
\end{brrule}

As mentioned above, if the branching rule cannot be applied anymore, then each valid improving swap contains at least two vertices.
Before applying the second branching rule, we perform the following preprocessing. 
First, we compute for each~$j\in[1,d]$ some minimum valid connected~$j$-improving~$(k-d+j)$-swap~$W_j$ for~$S$ in~$G$ if there is any. 
We call a collection~$\Wsf:=\{W_j\}_1^d$ of such swaps a~\emph{swap family}.
The idea of the branching rule is the following: for each~$j \in [1,d]$, either the~$d$-improving swap contains~$W_j$, some neighbor of~$W_j$, or no connected component that is exactly~$j$-improving.
Here, it suffices to consider only~$(k-d+j)$-swaps and not general~$k$-swaps, as a~$d$-improving~$k$-swap containing a connected component~$W_j'$ that is exactly~$j$-improving needs at least~$d-j$ vertices outside of~$W_j'$ to be~$d$-improving. That is, $W_j'$ can have size at most~$k-d+j$.
The correctness of the branching rule comes from the following observation:
Consider some valid minimum \goodswap~$W$ for~$I$ and let~$C$ be a connected component in~$G[W]$ with the minimal improvement.
Let~$\ell := \imp(C)$ and let~$W'= (W \setminus C) \cup W_\ell$. 
Since~$W_\ell$ contains at most~$|C|$ vertices, we have~$|W'| \leq k$. 
The resulting swap~$W'$ is good if~$W \cap N[W_\ell] = \emptyset$.

We now formally prove that this intuition is correct.

\begin{proposition}\label{correctness rule small swaps}
Let~$I=(G=(V,E),S,k,d)$ be a yes-instance of~\GLVC and, let~$\Wsf$ be a swap family.
There is a \goodswap~$W$ for~$I$ such that \textnormal{(i)}~$W$ is connected or \textnormal{(ii)} there is some~$W_j\in \Wsf$ such that~$j\le d/2$ and~$W_j \subseteq W$ or~$W \cap N(W_j) \neq \emptyset$. 
\end{proposition}
\begin{proof}
Let~$W$ be a minimum \goodswap for~$I$.
The improvement of~$W$ is exactly~$d$, as otherwise removing an arbitrary vertex of~$S$ from~$W$ yields a \goodswap as well.
Suppose that~$W$ is not connected and that $W_j \not\subseteq W$ and~$W \cap N(W_j) = \emptyset$ for every~$W_j\in \Wsf$, as otherwise the statement already holds.
Let~$C$ be a connected component in~$G[W]$ that minimizes~$\imp(C)$ among all connected components of~$G[W]$, that is, $C$~is the connected swap of~$W$ with the smallest improvement.
Let~$\ell :=   \imp(C)$.
Since~$W$ is not connected and has improvement exactly~$d$, the improvement~$\ell$ of~$C$ is at most~$\lfloor \frac{d}{2}\rfloor$.
Note that~$\ell \geq 1$, as otherwise $W$ is not minimum.

Note that~$W' :=   W \setminus C$ is a valid~$(d-\ell)$-improving~$(k-|C|)$-swap for~$I$ and~$|C| < k$.
Since~$W$ has size at most~$k$ and is~$d$-improving, $W\setminus S$ is~$(d-\ell)$-improving.
This implies that~$|C| \leq k-d+\ell$.
Recall that~$W_\ell$ is some minimum valid connected~$\ell$-improving~$(k-d+\ell)$-swap for~$S$ in~$G$.
Hence, $|C| \geq |W_\ell|$. 
In the following, we show that~$W' \cup W_\ell = (W\setminus C) \cup W_\ell$ is a \goodswap for~$I$. 
To this end, we first show that~$W$ contains no vertex of~$W_\ell$.

Assume towards a contradiction that~$W\cap W_\ell$ is nonempty and let~$C_W$ be some connected component of~$G[W]$ that contains at least one vertex of~$W_\ell$.
Recall that by assumption~$W_\ell \not\subseteq W$ and~$N(W_\ell) \cap W = \emptyset$.
Hence, $C_W$ is a proper subset of~$W_\ell$.
Since~$W_\ell$ is a smallest valid~$\ell$-improving~$k$-swap for~$S$ in~$G$, $\imp(C_W) < \imp(W_\ell) = \ell = \imp(C)$.
This contradicts the fact that~$C$ is a connected component of~$G[W]$ of minimal improvement.    
Consequently, $W$~contains no vertex of~$W_\ell$.

We set~$W^* :=   W' \cup W_\ell$.
Note that~$W^*$ is a~$d$-improving~$k$-swap for~$S$ in~$G$.
It remains to show that~$W^*$ is valid.
Since~$W'$ and~$W_\ell$ are both valid, it follows that~$W^*$ is valid if~$W' \cap N(W_\ell \cap S) \cap S = \emptyset$.
By assumption, this is the case.
\end{proof}

\Cref{correctness rule isolated} now directly implies the correctness the following branching rule.

\begin{brrule}\label{rule:swaps no weight}
  Let~$I=(G=(V,E),S,k,d)$ be an instance of~\GLVC such that there is no connected \goodswap for~$I$ and let~$\Wsf=\{W_j\}_1^d$ be a given swap family.
For each~$W_j\in \Wsf$, $j\le d/2$, and each swap~$W \in \{\{w\} \mid w \in N(W_j)\} \cup \{W_j\}$, branch into the case of swapping~$W$.
\end{brrule}

Since this branching relies on knowing a swap family, we next present an algorithm to efficiently find such a swap family using the algorithm of Katzmann~and~Komusiewicz~\cite{KK17} as a subroutine.

\begin{proposition}\label{running time rule small swaps}
Let~$I=(G=(V,E),S,k,d)$ be an instance of~\GLVC such that for each vertex~$v\in S$, $N(v) \setminus S \neq \emptyset$.
One can compute a swap family~$\Wsf :=\{W_j\}_1^d$ in time~$\Oh(2^k\cdot(\Delta-1)^{(k+d)/2}\cdot k^3 \cdot n)$. 
\end{proposition}
The proof of~\Cref{running time rule small swaps} relies on the following lemma.

\begin{lemma}[\cite{KK17}]\label{lem:bipartite}
Let~$G=(B\cup W, E)$ be a bipartite and connected graph with partite sets~$B$ and~$W$ where~$|B|>|W|$. 
Then, there is a vertex set~$B' \subseteq  B$, such that~$|B'| =
|W| + 1$ and~$B'\cup W$ is connected.
\end{lemma}
Since there is no published proof for this lemma, we provide one for the sake of completeness.
\begin{proof}
We prove the statement by showing that there is a vertex~$v\in B$ such that~$G[(B\setminus \{v\})\cup W]$ is connected.
Recursively applying this argument then implies the desired statement.
Let~$T=(B\cup W, E')$ be an arbitrary spanning tree of~$G$, which is rooted at some vertex~$w\in W$.
We show that there is a vertex~$v\in B$ which is a leaf of~$T$, which then implies that~$T[(B\setminus \{v\})\cup W]$ is a spanning tree of~$G[(B\setminus \{v\})\cup W]$.
Let~$d$ be the depth of~$T$ and for each~$i\in [1,d]$, let~$L_i$ be the vertices of the~$i$th level of~$T$.
Since~$G$ is bipartite, $L_i\subseteq B$ if~$i$ is odd and~$L_i\subseteq W$ if~$i$ is even.
If~$d$ is odd, then each vertex of~$L_d\subseteq B$ is a leaf of~$T$.
Otherwise, assume~$d$ is even.
For each odd~$i\in[1,d]$, $L_i$ contains a leaf of~$T$ or~$|L_{i+1}| \geq |L_i|$.
Since~$$\sum_{\text{odd}~i\in[1,d]}|L_i|  = |B| > |W| > \sum_{\text{even}~i\in[1,d]}|L_i|,$$ there is at least one odd~$i\in[1,d]$ with~$|L_{i+1}| < |L_{i}|$, which implies that~$L_i\subseteq W$ contains a leaf of~$T$.
\end{proof}

We are now able to prove~\Cref{running time rule small swaps}.

\begin{proof}[Proof of~\Cref{running time rule small swaps}]
Due to~\Cref{lem:bipartite}, for each valid connected improving swap~$W$ for~$S$ in~$G$, there is an independent set~$J \subseteq W \cap S$ such that~$W' = W \setminus J$ is a valid connected swap for~$S$ in~$G$ with~$|W' \cap S| = |W' \setminus S| + 1$. 
Katzmann~and~Komusiewicz~\cite{KK17} presented an algorithm~$\mathcal{A}$ to enumerate all valid connected~$(k-d + 1)$-swaps~$W'$ for~$S$ in~$G$ with~$|W' \cap S| = |W' \setminus S| + 1$ in $\Oh(2^k \cdot (\Delta(G) - 1)^{(k-d)/2}\cdot k^2 \cdot n)$~time.
Our algorithm relies on the following idea: 
Let~$W_j$ be a minimum valid connected~$j$-improving~$(k-d+j)$-swap for~$S$ in~$G$ for some~$j \in [1,d]$.
Then, \Cref{lem:bipartite} implies that there is an independent set~$J_j\subseteq W_j\cap S$ of size~$j-1$, such that~$W_j\setminus J_j$ is a valid~$1$-improving~$(k-d+1)$-swap for~$S$ in~$G$.
Hence, to find the swap~$W_j$, the intuitive idea behind our algorithm is to enumerate all valid connected~$(k-d + 1)$-swaps~$W'$ for~$S$ in~$G$ with~$|W' \cap S| = |W' \setminus S| + 1$ by using algorithm~$\mathcal{A}$, and afterwards finding an independent set of size~$j-1$ in~$(V\cap S)\setminus W'$ that can extend~$W'$. 

We initialize~$W_j$ with~$\bot$ for each~$j\in [1,d]$. 
For each swap~$W'$ outputted by~$\mathcal{A}$, we first compute an auxiliary graph~$G'_{W'}$. Then, we find a largest independent set~$J$ in~$G'_{W'}$ subject to the constraint~$|J|\le d-1$.
Intuitively, the graph~$G'_{W'}$ is the subgraph of~$G$ induced by exactly those vertices of~$S$ that can be swapped together with~$W'$ to still obtain a valid connected swap.
That is, $G'_{W'} :=   G[\{v\in S\setminus W' \mid N(v) \subseteq S \oplus W', v\in N(W'\setminus S)\}]$.
Since 
\begin{itemize}
\item $W'$ contains at least one vertex of~$V\setminus S$ (as $W'$ is a valid swap and each vertex of~$S$ has at least one neighbor in~$V\setminus S$ by assumption), 
\item $G'_{W'}$ contains no vertex of~$W'$, and
\item each vertex of~$G'_{W'}$ is a neighbor of some vertex of~$W'\setminus S$ in~$G$,
\end{itemize}
the graph~$G'_{W'}$ has a maximum degree of at most~$\Delta - 1$.
This implies that we can find the set~$J$ in $\Oh((\Delta - 1)^{d-1} \cdot (\Delta - 1)\cdot k)$~time by a standard search-tree algorithm, since~$G'_{W'}$ contains at most~$(\Delta - 1)\cdot k$ vertices.
Note that for each subset~$J'$ of~$J$, $W' \cup J'$ is a valid connected~$(|J'| + 1)$-improving~$k$-swap for~$S$ in~$G$.
For each~$r\in [0,|J|]$, let~$J_r$ be an arbitrary subset of~$J$ of size~$r$ and update~$W_{r+1}$ to~$W' \cup J_r$ if~$|W' \cup J_r| < |W_{r+1}|$.

Next, we show that the algorithm is correct.
Let~$j\in [1,d]$ such that a minimum~$j$-improving~$(k-d+j)$-swap~$W^*_j$ for~$S$ in~$G$ exists.
We show that the set~$W_j$ computed by our algorithm has size~$|W^*_j|$.
Since~$W^*_j$ is a~$j$-improving~$(k-d+j)$-swap and~$W^*_j$ is minimum, there is an independent set~$J_j\subseteq W^*_j \cap S$ of size~$j-1$ such that~$W' :=   W^*_j \setminus J_j$ is connected due to~\Cref{lem:bipartite}.  
Hence, $W'$ is a valid connected~$(k-d+1)$-swap for~$S$ in~$G$.
As a consequence, the algorithm~$\mathcal{A}$ outputs~$W'$.
Moreover, since~$W^*_j$ is valid, for each vertex~$v\in J_j$, $N(v)\subseteq S \oplus W'$.
Hence, $J_j$ is an independent set of size~$j-1$ in~$G'_{W'}$.
Thus, when considering~$W'$, our algorithm finds some independent set~$J'$ of size at least~$j-1$ in~$G'_{W'}$ and either updates~$W_j$ to~$W'\cup J_j'$ for some subset~$J_j'\subseteq J'$ of size~$j-1$ or~$W_j$ already has size~$|W^*_j|$.
Hence, after our algorithm considered the swap~$W'$, $|W_j|$ is exactly~$|W'\cup J_j'| = |W^*_j|$ since~$W^*$ is minimum.
Moreover, since~$G'_{W'}$ contains only vertices that are adjacent to at least one vertex of~$W'\setminus S$, $W'\cup J_j'$ is connected.

Hence, for all~$j\in[1,d]$, if it exists, we can find some minimum valid connected~$j$-improving~$(k-d+j)$-swap~$W_j$ for~$I$ where~$|W_j \setminus S| \leq (k-d)/2$ in total time~$\Oh(2^k\cdot\Delta^{(k+d)/2}\cdot k^3\cdot n)$.
\end{proof}

With the above two branching rules and the algorithm to compute a swap family, we are now able to prove~\Cref{glvc delta k d}.
\begin{proof}[Proof of~\Cref{glvc delta k d}.]
Let~$I=(G=(V,E), S,k,d)$ be an instance of~\GLVC.
If~$k+d \leq 4$, we can solve~$I$ in $\Oh(n+m)$~time due to~\Cref{lem: polytime k d small}.
If~$\Delta(G) = 2$, then we can compute a nice tree decomposition of~$G$ of width at most~$2$ in $\Oh(n+m)$~time and afterwards solve~$I$ in $\Oh((2^{(k+d)/2} + k^2)\cdot n + m)$~time due to~\Cref{thm: tw k d}.
Since~$d \leq k$, this is $\Oh(k!\cdot n + m)$~time.
Hence, we can assume in the following, that~$k+d > 4$ and~$\Delta(G) \geq 3$.

First, check in $\Oh(n+m)$~time if there is a vertex~$v\in S$ with~$N(v) \subseteq S$.
If this is the case, apply \Cref{rule:deg one}.
Due to~\Cref{correctness rule isolated}, this is correct.
If there is no vertex~$v\in S$ with~$N(v) \subseteq S$, compute a swap family~$\Wsf:=\{W_j\}_1^d$.
Due to \Cref{running time rule small swaps}, this can be done in~$\Oh(2^k\cdot(\Delta-1)^{(k+d)/2}\cdot k^3\cdot n)$~time.
If~$W_d\in \Wsf$, answer yes.
Otherwise, apply~\Cref{rule:swaps no weight}.
Due to~\Cref{correctness rule small swaps}, this is correct.

It remains to show the total running time.
Let~$\mathcal{I}_1$ denote the set of swap-instances considered during one application of~\Cref{rule:deg one} and let~$\mathcal{I}_2$ denote the set of swap-instances considered during one application of~\Cref{rule:swaps no weight}. 
We first analyze~$\mathcal{I}_1$ and~$\mathcal{I}_2$.
Note that~$\mathcal{I}_1$ consists of at most~$\binom{\Delta}{2}$ instances where~$k' \leq k - 2$ and~$k'+d' \leq k + d - 4$ and exactly one instance where~$k' = k - 1$ and~$k'+d' = k + d - 2$.
Since~$\Delta \geq 3$, $\binom{\Delta}{2} \leq (\Delta - 1)^2$.

Hence, if~\Cref{rule:deg one} is applicable, the recurrence for the running time is
$$
T(k,k+d) \leq (\Delta - 1)^2\cdot T(k-2, k+d-4) + T(k-1, k+d-2) + \Oh(n+m).
$$

Under the assumption that~\Cref{rule:deg one} is not applicable, we bound the number of instances in~$\mathcal{I}_2$ and show that $k' \leq k - 2$ and~$k'+d' \leq k+d - 2$ for each instance~$(G',S',k',d')\in \mathcal{I}_2$.
Let~$W_\ell$ be the largest swap in~$\Wsf$.
Note that~$\ell \leq  \frac{d}{2} < d \leq k$.
Hence, each swap~$W_j\in \Wsf$ has size at most~$k-1$.
Further, since~$W_j$ is connected, each vertex in~$W_j$ has at least one neighbor in~$W_j$ and thus at most~$\Delta - 1$ neighbors outside of~$W_j$.
Altogether, 
$$|\mathcal{I}_2| \leq \frac{d}{2}\cdot ((k-1) \cdot (\Delta - 1) + 1) \leq \frac{k^2}{2}\cdot (\Delta-1),$$
since~$\Delta \geq 2$.
By the assumption that~\Cref{rule:deg one} is not applicable, there is no vertex~$v\in S$ with~$N(v) \subseteq S$.
This implies that~$k' \leq k - 2$.
Further, since each considered swap contains at least one vertex of~$S$, $k'+d' \leq k+d - 2$ for each instance of~$\mathcal{I}_2$.

Hence, if~\Cref{rule:deg one} is not applicable and~\Cref{rule:swaps no weight} is applicable, the recurrence for the running time is
$$
T(k,k+d) \leq \frac{k^2}{2}\cdot (\Delta - 1) \cdot T(k-2, k+d-2) + \Oh(2^{k}\cdot (\Delta-1)^{(k+d)/2}\cdot k^2 \cdot n).
$$

We now show by induction over~$k$ that there is a constant~$C$ such that for each~$d\in[1,k]$, $T(k, k+d)\leq C \cdot k! \cdot (\Delta-1)^{(k+d)/2}\cdot n$.

By the above, there is a constant~$c$ such that for each~$k\in \mathbb{N}$ and each~$d\in[1,k]$
\begin{enumerate}[label=\textbf{(\arabic*)}]
\item\label{enuma} \GLVC can be solved in $c\cdot (n+m)$~time if~$k < 3$, 
\item\label{enumb} $T(k,k+d) \leq (\Delta-1)^2\cdot T(k-2, k+d-4) + T(k-1, k+d-2) + c\cdot (n+m)$ if there is some vertex~$v\in S$ with~$N(v) \subseteq S$, and
\item\label{enumc} $T(k,k+d) \leq \frac{k^2}{2}\cdot (\Delta - 1) \cdot T(k-2, k+d-2) + c\cdot 2^{k}\cdot (\Delta-1)^{(k+d)/2}\cdot k^3 \cdot n$ if there is no vertex~$v\in S$ with~$N(v) \subseteq S$.
\end{enumerate}

We set~$C :=   128c$.
For the base case~$k < 3$ the statement now holds directly due to~\ref{enuma}.
As the inductive step, we show that the statement holds for~$k$ if it holds for all~$k'\in [1,k-1]$.

Suppose that there is some vertex~$v\in S$ with~$N(v) \subseteq S$. 
By the induction hypothesis and due to~\ref{enumb}, 
\begin{align*} 
T(k,k+d)& \leq&&(\Delta-1)^2\cdot T(k-2, k+d-4) + T(k-1, k+d-2) + c\cdot (n+m)\\ 
&\leq&&C \cdot (\Delta-1)^2 \cdot (k-2)! \cdot (\Delta-1)^{(k+d)/2-2}\cdot n&\\
&&&+~C \cdot (k-1)! \cdot (\Delta-1)^{(k+d)/2-1}\cdot n + c\cdot (n+m)&\\
&=&& C \cdot (k-2)! \cdot (\Delta-1)^{(k+d)/2-1} \cdot ((\Delta-1) + (k-1))\cdot n + c \cdot (n+m) &\\
&\overset{(*)}{\leq}&& C \cdot (k-1)! \cdot (\Delta-1)^{(k+d)/2}\cdot n + c\cdot (n+m)&\\
&\overset{(**)}{\leq}&& C \cdot k! \cdot (\Delta-1)^{(k+d)/2}\cdot n&
\end{align*} 
Inequality~$(*)$ holds, since~$(\Delta-1) + (k-1) \leq (\Delta-1) \cdot (k-1)$ for all~$k\geq 3$ and~$\Delta \geq 3$.
Inequality~$(**)$ holds, since~$2c \leq C$, $k\geq 3$, and~$m\leq (\Delta-1)\cdot n$.

Suppose that there is no vertex~$v\in S$ with~$N(v) \subseteq S$. 
By the induction hypothesis and due to~\ref{enumc}, 
\begin{align*} 
T(k,k+d)& \leq&& \frac{k^2}{2}\cdot (\Delta - 1) \cdot T(k-2, k+d-2) + c\cdot 2^{k}\cdot (\Delta-1)^{(k+d)/2}\cdot k^3 \cdot n\\
&\leq&&C \cdot \frac{k^2}{2}(\Delta-1) \cdot (k-2)! \cdot (\Delta-1)^{(k+d)/2-1}\cdot n\\
&&&+~c \cdot 2^{k} \cdot (\Delta-1)^{(k+d)/2}\cdot k^3 \cdot n\\
&\leq&& C \cdot \left(\frac{k^2}{2} \cdot (k-2)! + \frac{1}{128} \cdot 2^{k} \cdot k^3\right) \cdot (\Delta-1)^{(k+d)/2}\cdot n\\
&\overset{(*)}{\leq} &&C \cdot k! \cdot (\Delta-1)^{(k+d)/2}\cdot n
\end{align*} 
Inequality~$(*)$ holds, since~$C = 128c$ and $\frac{k^2}{2} \cdot (k-2)! + \frac{1}{128} \cdot 2^k \cdot k^3 \leq k!$ for all~$k\geq 3$.

Hence, the whole algorithm runs in $\Oh(k! \cdot (\Delta-1)^{(k+d)/2}\cdot n)$~time.
\end{proof}

Note that the above running time is (besides the change from the~$2^k$ factor to a~$k!$ factor in the running time) a direct generalization of the previous best algorithm for~\LVC which runs in~$\Oh(2^k \cdot (\Delta-1)^{k/2}\cdot k\cdot n)$~time~\cite{KK17} to~\GLVC.

\subsection{Parameterizing Weighted Gap Local Search by Maximum Degree}\label{sec:weighted-degree}

Next, we consider~\GLWVClong when parameterized by~$\Delta(G)$ plus~$k$.

\begin{proposition}\label{running time rule small swaps weighted}
Let~$I=(G=(V,E),\omega,S,k,d)$ be an instance of~\GLWVC.
One can enumerate all valid connected~$k$-swaps for~$S$ in~$G$ in $\Oh(2^k \cdot (\Delta-1)^{k}\cdot k^3 \cdot n)$~time. 
\end{proposition}
\begin{proof}[Proof sketch]
We adapt an algorithm of Katzmann~and~Komusiewicz~\cite{KK17}. 
Note that this algorithm does not consider the weights of vertices and only enumerates~$k$-swaps~$W$ with~$|W\cap S| = |W\setminus S| + 1$.
Unfortunately, in weighted graphs, a valid improving swap~$W$ does not have to fulfill~$|W\cap S| = |W\setminus S| + 1$.
To cope with this, we use the following algorithm.

The idea of our algorithm is to enumerate, for each vertex~$v\in S$, the valid connected $k$-swaps~$W$ containing~$v$. 
In the following, call the vertices in the vertex cover~$S$~\emph{black} and the vertices of the independent set~$V\setminus S$~\emph{white}. 
Recall that if a black vertex~$u$ is contained in a swap~$W$, then no black neighbors of~$u$ may be included in~$W$ and all white neighbors of~$u$ must be included in~$W$. 
In other words, white vertices are added automatically, when they are neighbors of the current swap. 
Now, to extend a current swap~$W$, we pick a white vertex~$u$ in~$W$ and choose 1) to exclude all neighbors of~$u$ outside of~$W$  to be swapped in all recursive calls or 2) to pick one of the at most~$\Delta-1$ neighbors of~$u$ outside of~$W$ to be added to~$W$. 
In Case~1), we ``finish'' a white vertex, in the at most~$\Delta-1$ subcases of Case~2), we add a black vertex. 
The total number of white and black vertices is at most~$k$, and thus we may abort the search at a search tree depth of~$k$. 
The total search tree size is thus~$\Oh(\Delta^k)=\Oh(2^k\cdot (\Delta-1)^k)$ for each of the~$\Oh(n)$ initial vertices~$v$. 
We omit the discussion of the further polynomial parts of the running time.        \end{proof}

Since any instance~$I$ of~\LWVC has a~\goodswap if and only if it has a \emph{connected}~\goodswap (see~\Cref{connected}), \Cref{running time rule small swaps weighted} implies that~\LWVC can be solved in the following running time.

\begin{corollary}
\LWVC can be solved in $\Oh(2^k\cdot(\Delta-1)^{k}\cdot k^3 \cdot n)$~time. 
\end{corollary}

To solve the gap version~\GLWVC we again encounter the problem that the sought \goodswap is not necessarily connected.
Hence, for~\GLWVC we show a related algorithm to the one we presented for~\GLVC using only one branching rule.
This rule is, more or less, an adaptation of~\Cref{rule:swaps no weight} to the weighted version.
Similar to the unweighted case, we want to compute a small collection~$\{W_j\}$ of representative improving swaps and branch to either take one of them or at least one neighbor of one of them.
As we aim for a running time depending on~$k$ and we are dealing with a weighted graph now, we cannot compute such a representative swap~$W_j$ for each~$j\in [1,d]$.
Instead, compared to the unweighted case, $W_j$ should now denote some valid connected improving~$j$-swap of largest improvement. 
That is, the subscript now indicates the size of the swap and not its improvement. 
As already mentioned, after finding such a swap for each~$j\in[1,k]$, we aim to branch into the cases of either swapping~$W_j$ or swapping some neighbor of~$W_j$.

Unfortunately, a result similar to~\Cref{correctness rule isolated} cannot be obtained, that is, in the weighted case, each~\goodswap might contain only one neighbor of~$W_1$, the best single-vertex swap as shown in \Cref{fig:w1}.
\begin{figure}[t]
  \centering
  \begin{tikzpicture}[yscale=.6, xscale = 1.5]
    \tikzstyle{knoten}=[circle,fill=white,draw=black,minimum size=7pt,inner sep=0pt]
    \tikzstyle{blocked}=[circle,fill=black,draw=black,minimum size=7pt,inner sep=0pt]

        \node[blocked] (u) [label=below:{$u$}, label=above:{$1$}] at (0,0) {};
    \node[blocked] (v) [label=below:{$v$}, label=above:{$3$}] at (1.5,0) {};
    \node[knoten] (w) [label=below:{$w$}, label=above:{$1$}] at (3,0) {};

        \draw[-, ultra thick] (u) to (v);
    \draw[-, ultra thick] (v) to (w);

        \node[left] at ($(u)!0.5!(v) + (-0.8,0)$) {$S$};

  \end{tikzpicture}

  \caption{The only~$2$-improving swap for~$S$ is~$\{v,w\}$. This swap avoids the only valid improving~$1$-swap~$\{u\}=W_1$ and contains only one neighbor of~$u$.}
\label{fig:w1}
\end{figure}
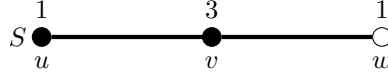
Hence, in the worst case, each of these branching cases decreases the parameter~$k$ only by one which would lead to a running time factor of~$(\Delta-1)^{2 k}$ instead of~$(\Delta-1)^{k}$.
Our goal is, thus, to decrease the number of cases in which the parameter is only decreased by one.
To this end, we analyze the swap~$W_1$ separately.
Let~$S_1 :=   \{v\in S \mid N(v)\subseteq S\}$ denote the set of vertices of improving~$1$-swaps for~$S$ in~$G$ and let~$v^*$ be the unique vertex of~$W_1$, that is, $v^*$ is a vertex of highest weight in~$S_1$. 
Consider a \goodswap~$W$ for~$I$.
Assume~$W$ contains some vertex~$w^*\neq v^*$ of~$S_1$, then either~$W\cap N[W_1] = W\cap N[v^*]$ is not empty or replacing~$w^*$ by $v^*$ also gives a \goodswap for~$I$. Consequently, after considering the cases where~$W$ contains some vertex of $N[W_1]$ we may assume that~$W$ contains no vertex of~$S_1$.
Hence, we can restrict our branching cases for~$j\geq 2$ to the ones in which we consider either swapping~$W_j$ or some neighbor~$w$ of~$W_j$ which is not contained in~$S_1$. This gives a swap of size at least 2 because~$w$ has a neighbor in~$V\setminus S$. 
Thus, we have~$|N[W_1]| \leq \Delta + 1$ cases where the parameter is decreased by only one and in all other cases, the parameter decrease is at least 2.  

To formulate the branching rule, for each~$i\in [1,k]$, let~$W_j$ denote some valid connected size-$j$ swap for~$I$ that has maximal improvement, if such a swap exists.
and call the collection~$\Wsf = \{W_j\}_1^k$ of these swaps a~\emph{weighted swap family}.
Recall that~$S_1 :=   \{v\in S \mid N(v)\subseteq S\}$ denotes the set of vertices of improving~$1$-swaps for~$S$ in~$G$. 

\begin{proposition}\label{correctness rule small swaps weighted}
Let~$I=(G=(V,E),\omega,S,k,d)$ be a yes-instance of~\GLWVC and let~$\{W_j\}_1^k$ be a weighted swap family.
There is a \goodswap~$W$ for~$I$ such that \textnormal{(i)}~$W$ is connected or \textnormal{(ii)}~$W \cap N[W_1] \neq \emptyset$ or \textnormal{(iii)} there is some~$W_j\in \Wsf$,~$2\le j\le {k}/{2},$ such that~$W_j \subseteq W$ or~$(W \cap N(W_j)) \setminus S_1 \neq \emptyset$.
\end{proposition}
\begin{proof}
Since~$I$ is a yes-instance of~\GLWVC, there is a minimum \goodswap~$W$ for~$I$.
Suppose that~$W$ is not connected, $W \cap N[W_1] = \emptyset$, and $W_j \not\subseteq W$ and~$W \cap N(W_j)\setminus S_1 = \emptyset$ for each~$W_j\in \Wsf$, $2\le j\le k/2$, as otherwise the statement already holds.
Let~$C$ be the smallest connected component in~$G[W]$ and let~$\ell := |C|$.
Since~$W$ is not connected, $\ell \leq \lfloor \frac{k}{2}\rfloor$.
Note that~$W' :=   W \setminus C$ is a valid~$(d-\imp(C))$-improving~$(k-\ell)$-swap for~$S$ in~$G$.
Since~$C$ is a connected swap of size exactly~$\ell$, the swap~$W_\ell$ exists. 
Recall that~$W_\ell$ is some valid connected swap of size exactly~$\ell$ for~$I$ that maximizes~$\imp(W_\ell)$.
Hence, the improvement~$\imp(W_\ell)$ of~$W_\ell$ is at least the improvement~$\imp(C)$ of~$C$. 
In the following, we show that~$W^* :=   W' \cup W_\ell$ is a \goodswap for~$I$.

Recall that~$W_\ell \not\subseteq W$ and~$(N(W_\ell) \cap W) \setminus S_1 = \emptyset$.
Further, since each vertex of~$S_1\cap W$ is isolated in~$G[W]$, $W_\ell \cap W$ is a proper subset of~$W_\ell$ and~$W_\ell$ is a valid swap for~$S$ in~$G$.
Hence, $W_\ell \cap W = \emptyset$, as otherwise~$C$ is not the smallest connected component in~$G[W]$ by the fact that~$W \cap N(W_\ell)\setminus S_1 = \emptyset$.
Note that~$W^*$ is a~$d$-improving~$k$-swap for~$S$ in~$G$.
It remains to show that~$W^*$ is valid.
Since~$W'$ and~$W_\ell$ are both valid, it follows that~$W^*$ is valid if~$W' \cap N(W_\ell \cap S) = \emptyset$.
By assumption, this is the case for~$\ell = 1$.
Moreover by assumption, $W' \cap N(W_\ell \cap S) \setminus S_1 = \emptyset$ if~$\ell > 1$.
Note that since every valid connected swap containing some vertex of~$S_1$ has size~$1$, $W' \cap  S_1 \neq \emptyset$ would then contradict the assumption that~$C$ is the smallest connected swap in~$G[W]$ with~$|C| = \ell > 1$. 
\end{proof}

Hence, we derive the following branching rule.

\begin{brrule}\label{rule:swaps with weight}
Let~$I=(G=(V,E),\omega,S,k,d)$ be an instance of~\GLWVC such that there is no connected \goodswap for~$I$.
Moreover, let~$\Wsf:=\{W_j\}_1^k$ be a given weighted swap family.
If~$W_1$ exists, branch into the case of swapping~$W$ for each swap~$W\in \{\{w\}\mid w\in N(W_1)\}$.
Additionally, for each~$W_j\in \Wsf$, $2\le j\le k/2$, branch into the case of swapping~$W$ for each swap~$W\in \{W_j\} \cup \{\{w\} \mid w \in N(W_j) \setminus S_1\}$.
\end{brrule}

With this branching rule, we can now present the algorithm for~\GLWVC.

\begin{theorem}\label{glwvc delta k d}
\GLWVC can be solved in~$\Oh(k!\cdot (\Delta-1)^{k}\cdot n)$ time.
\end{theorem}
\begin{proof}
Let~$I=(G=(V,E), \omega, S,k,d)$ be an instance of~\GLVC.
If~$k \leq 2$, then we can solve~$I$ in $\Oh(n+m)$~time due to~\Cref{lem: polytime k small}.
If~$\Delta(G) = 2$, then the treewidth of~$G$ is at most~$2$ and we can compute a nice tree decomposition of~$G$ of width at most~$2$ in $\Oh(n+m)$~time and afterwards solve~$I$ in $\Oh((2^{k}+ k^2)\cdot n + m) \subseteq \Oh(k!\cdot n + m)$~time due to~\Cref{thm width and k fpt}.
Hence, we can assume in the following, that~$k \geq 3$ and~$\Delta(G) \geq 3$.

The algorithm now works as follows.
First, compute a weighted swap family~$\Wsf:=\{W_j\}_1^k$.
Due to~\Cref{running time rule small swaps weighted}, this can be done in~$\Oh(2^k\cdot(\Delta-1)^{k}\cdot k^3\cdot n)$~time.
Now, if there is some~$W_j\in \Wsf$ such that~$\imp(W_k) \geq d$, then~$I$ is a yes-instance of~\GLWVC.
Otherwise, there is no connected \goodswap~$W$ for~$I$.
Compute the set~$S_1 :=   \{v\in S \mid N(v)\subseteq S\}$ of possible improving swaps of size~$1$ and apply~\Cref{rule:swaps with weight}.
Due to~\Cref{correctness rule small swaps weighted}, this is correct.

It remains to show the running time.
Consider an application of application of~\Cref{rule:swaps with weight} and let~$\mathcal{I}_1$ denote the set of constructed swap-instances where a vertex of~$N[W_1]$ is swapped and let~$\mathcal{I}_{>1}$ denote the remaining constructed swap-instances. 
That is, 
$$\mathcal{I}_1 :=   \begin{cases}
\{\swap(I,\{w\}) \mid w \in N[W_1]\}& W_1~\text{exists},\\
\emptyset& \text{otherwise,}
\end{cases}$$ and $$\mathcal{I}_{>1} :=   \bigcup_{\substack{W_j\in \Wsf\\j\in [2,  \lfloor \frac{k}{2}\rfloor]}}\{\swap(I,W_j)\} \cup \{\swap(I,\{w\}) \mid w \in N(W_j) \setminus S_1\}.$$

Note that~$\mathcal{I}_1$ has size at most~$\Delta + 1$ and~$k' \leq k-1$ for each instance~$I'\in \mathcal{I}_1$.
Next, we bound the size of~$\mathcal{I}_{>1}$.
Since each considered~$W_j$ has size~$j$ and each vertex~$v\in W_j$ has at most~$\Delta - 1$ neighbors outside of~$W_j$, we can upper-bound the size of~$|\mathcal{I}_{>1}|$ as 
$$|\mathcal{I}_{>1}| \leq \sum_{j=2}^{\lfloor\frac{k}{2}\rfloor} (j \cdot (\Delta - 1) + 1) \leq \frac{\frac{k}{2} \cdot (\frac{k}{2}+1)}{2}\cdot (\Delta - 1)  + \frac{k}{2} - \Delta = \frac{k^2+2k}{8} \cdot (\Delta - 1) + \frac{k}{2} - \Delta.$$
Note that for each instance of~$\mathcal{I}_{>1}$, we have~$k' \leq k - 2$ because $|W_j| \geq 2$ and~$N(w)\setminus S \neq \emptyset$ for each vertex~$w\in N(W_j) \setminus S_1$.

Hence, the recurrence for the running time is
$$
T(k) \leq \left(\frac{k^2+2k}{8} \cdot (\Delta - 1) + \frac{k}{2} - \Delta\right) \cdot T(k-2) + (\Delta + 1)\cdot T(k-1) + \Oh(2^{k}\cdot (\Delta-1)^{k}\cdot k^2 \cdot n).
$$

We show by induction over~$k$ that there is a constant~$C$ such that~$T(k)\leq C \cdot k! \cdot (\Delta-1)^{k}\cdot n$.

First, observe that there is a constant~$c$ such that
\begin{enumerate}[label=\textbf{(\arabic*)}]
\item\label{enumaweight} \GLWVC can be solved in $c\cdot (n+m)$~time if~$k < 3$, 
\item\label{enumbweight} $T(k) \leq (\frac{k^2+2k}{8} \cdot (\Delta - 1) + \frac{k}{2} - \Delta) \cdot T(k-2) + (\Delta + 1)\cdot T(k-1)  + c\cdot 2^{k}\cdot (\Delta-1)^{k}\cdot k^3 \cdot n$.
\end{enumerate}

We set~$C :=   700 c$.
For the base case of~$k < 3$, the statement now holds directly due to~\ref{enumaweight}.
As the inductive step, we show that the statement holds for~$k$ if it holds for~$k-j$ for each~$j\in [1,k-1]$.

By the induction hypothesis and due to~\ref{enumbweight}, 
\begin{align*} 
T(k)& \leq&& \left(\frac{k^2+2k}{8} \cdot (\Delta - 1) + \frac{k}{2} - \Delta\right) \cdot T(k-2) \\
&&&+ (\Delta + 1)\cdot T(k-1)  + c\cdot 2^{k}\cdot (\Delta-1)^{k}\cdot k^3 \cdot n\\
&\leq&& C \cdot\left(\frac{k^2+2k}{8} \cdot (\Delta - 1) + \frac{k}{2} - \Delta\right) \cdot (k-2)! \cdot (\Delta-1)^{k-2}\cdot n\\
&&&+ C \cdot (\Delta + 1)\cdot (k-1)! \cdot (\Delta-1)^{k-1}\cdot n  +  c\cdot 2^{k}\cdot (\Delta-1)^{k}\cdot k^3 \cdot n \\
&=&& C \cdot (k-2)! \cdot (\Delta-1)^{k-2} \cdot \left(\frac{k^2+2k}{8} \cdot (\Delta - 1) + \frac{k}{2} - \Delta + (k-1)\cdot (\Delta^2 - 1)\right) \cdot n\\
&&&+ c\cdot 2^{k}\cdot (\Delta-1)^{k}\cdot k^3 \cdot n \\
&\overset{(*)}{\leq}&& C \cdot (k-2)! \cdot (\Delta-1)^{k} \cdot (k\cdot (k-1) - 1) \cdot n + c\cdot 2^{k}\cdot (\Delta-1)^{k}\cdot k^3 \cdot n \\
&\leq&& C \cdot k! \cdot (\Delta-1)^{k} \cdot n - C \cdot (k-2)! \cdot (\Delta-1)^{k}  \cdot n + c\cdot 2^{k}\cdot (\Delta-1)^{k}\cdot k^3 \cdot n \\
&\overset{(**)}{\leq}&& C \cdot k! \cdot (\Delta-1)^{k}\cdot n
\end{align*} 
Inequality~$(*)$ holds, since~$\frac{k^2+2k}{8} \cdot (\Delta - 1) + \frac{k}{2} - \Delta + (k-1)\cdot (\Delta^2 - 1) \leq (\Delta-1)^2\cdot (k \cdot (k-1) - 1)$ for all~$k \geq 3$ and~$\Delta \geq 3$.
Inequality~$(**)$ holds, since~$C = 700c$ and~$(k-2)! \geq \frac{1}{700} \cdot 2^k \cdot k^3$.

Hence, the whole algorithm runs in $\Oh(k! \cdot (\Delta-1)^{k}\cdot n)$~time.
\end{proof}

\subsection{Parameterizing Weighted Gap Local Search by $h$-Index}
Finally, we show that we can replace~$\Delta(G)$ in the above running time by the~$h$-index of~$G$. 
Recall that the~\emph{$h$-index}~$h(G)$ of a graph~$G$ is the largest integer such that~$G$ contains at least~$h(G)$ vertices of degree at least~$h(G)$. 
The idea behind this algorithm is to branch on all possibilities on how a hypothetical improving swap may intersect the set of high-degree vertices.
For each of these potential intersections~$X$, we compute the corresponding swap-instance and solve it with the help of~\Cref{glwvc delta k d} after removing the remaining high-degree vertices.
Intuitively, we want to avoid all possible valid swaps that contain any high-degree vertex outside of~$X$.
For each such  vertex~$v$, we have to consider two cases:
If~$v$ is contained in~$S$, we can simply remove~$v$.
Otherwise, if~$v\notin S$, we additionally have to remove each neighbor of~$v$ because each valid swap~$W$ that swaps any neighbor of~$v$ out of the vertex cover~$S$ has to also swap~$v$ into the vertex cover.

Based on this observation, we define for a given vertex set~$V' \subseteq V$ the~\emph{exclusion instance~$I'$ for~$V'$ and~$I$} as the instance of~\GLWVC, where all vertices of~$V'\cup N(V'\setminus S)$ are removed from~$I$. 
Formally, let~$I = (G,\omega,S,k,d)$ be an instance of~\GLWVC, then the exclusion instance~$I'$ of~\GLWVC for~$V'$ and~$I$ is defined as~$I' :=  (G',\omega,S', k,d)$, where~$G' :=   G - (V' \cup N(V'\setminus S))$ and~$S':=   S \cap V(G')$.
By the above argumentation, we derive the following property for exclusion instances.
\begin{lemma}\label{lem: remove subset}
Let~$I=(G=(V,E),\omega, S,k,d)$ be an instance of~\GLWVC and let~$V'\subseteq V$ be a set of vertices.
There is a \goodswap~$W$ for~$I$ that avoids~$V'$ if and only if the exclusion instance of~$V'$ and~$I$ is a yes-instance of~\GLWVC.
\end{lemma}
\begin{proof}
Let~$I'=(G',S',k,d)$ be the exclusion instance of~$V'$ and~$I$.
Let~$W$ be a \goodswap for~$I$ that avoids~$V'$.
Since~$W$ is valid, $W$ contains no vertex of~$N(V'\setminus S)$.
Hence, $W$ is a \goodswap for~$I'$.

Let~$W$ be a \goodswap for~$I'$.
By definition of the exclusion instance~$I'$, $W$ is also a~$d$-improving~$k$-swap for~$S$ in~$G$ which avoids~$V'$.
It remains to show that~$S\oplus W$ is a vertex cover of~$G$.
This is the case if no vertex of~$V'\setminus S$ is adjacent to some vertex in~$W$.
By construction of~$I'$, no vertex in~$V(G')$ has a neighbor in~$V'\setminus S$.
Hence~$W$ is a \goodswap for~$I$ that avoids~$V'$.
\end{proof}

Based on the concept of exclusion instances and due to~\Cref{glwvc delta k d}, we can now present the algorithm for~\GLWVC parameterized by~$h(G)$ plus~$k$.

\begin{theorem}\label{thm: h index algo}
\GLWVC can be solved in $\Oh(k!\cdot (h-1)^{k}\cdot n)$~time.
\end{theorem}
\begin{proof}
Let~$I=(G=(V,E),\omega,S,k,d)$ be an instance of~\GLWVC.
If~$k\leq 2$, the statement holds due to~\Cref{lem: polytime k small}.
Otherwise, we compute the~$h$-index~$h(G)$ of~$G$ and the set of vertices~$H$ with degree at least~$h(G)+1$ in $\Oh(n+m)$~time.
If~$h(G) = 2$, then the treewidth of~$G$ is at most~$4$ and we can compute a nice tree decomposition of~$G$ of width at most~$4$ in $\Oh(n+m)$~time and afterwards solve~$I$ in $\Oh((4^{k}+ k^2)\cdot n + m) \subseteq \Oh(k!\cdot n + m)$~time due to~\Cref{thm width and k fpt}.
Hence, we assume in the following that~$k \geq 3$ and~$h(G) \geq 3$.
Since~$h(G)$ is the~$h$-index of~$G$, $H$ contains at most~$h(G)$ vertices.
The idea is to consider all possibilities for the intersection of a \goodswap~$W$ with~$H$ and afterwards solve the resulting swap-instance with the algorithm of~\Cref{glwvc delta k d}.
To obtain a linear running time for instances where the~$h$-index and~$k$ are both constants, we initially compute the adjacency matrix of~$G[H]$ in $\Oh(h(G)^2 + m)$~time. 
For each~$W_H \subseteq H$ of size at most~$k$, we check whether~$W_H \cap S$ is an independent set and, if this is the case, compute the swap-instance~$I' :=   (G',\omega',S',k',d') :=   \swap(I, W_H)$ in $\Oh(k^2)$~time, if~$I'$ is a yes-instance of~\GLWVC, answer yes.

Note that if~$W_H$ has size exactly~$k$, then~$k' \leq 0$, which makes $I'$ a trivial instance of~\GLWVC that can be solved in $\Oh(1)$~time.

Otherwise, if~$W_H$ has size less than~$k$, we solve~$I'$ as follows.
Since we search for a \goodswap~$W$ for~$I$ with~$W \cap H = W_H$, no other vertex of~$H$ is contained in~$W$.
Hence, we can ignore choices of~$W_H$, where~$N(W_H\cap S)\cap H \not\subseteq W_H$.
For each remaining choice of~$W_H$, we can compute the exclusion instance~$I_{W_H}$ of~$H\cap V(G')$ and~$I'$ in $\Oh(h(G) \cdot n)$~time. 
Recall that due to~\Cref{lem: remove subset}, $I_{W_H}$ is a yes-instance of~\GLWVC if and only if there is a \goodswap for~$S'$ that avoids~$H\cap V(G')$.
Let~$G_{W_H}$ denote the graph of~$I_{W_H}$.
Since~$G_{W_H}$ is a subgraph of~$G$ and contains only vertices of~$V\setminus H$, $G_{W_H}$ has a maximum degree of~$h(G)$.
Moreover, note that~$k' \leq k - |W_H|$.
Consequently, $I_{W_H}$ can be solved in $\Oh((k-|W_H|)! \cdot (h(G) - 1)^{k-|W_H|} \cdot n)$~time due to~\Cref{glwvc delta k d}.
Note that this dominates the running time to compute~$I_{W_H}$ since~$|W_H| < k$.
 
The total running time of the algorithm can be upper-bounded by 
\begin{align*}
\Oh(n+m + h(G)^2) + &\Oh((h(G)-1)^{k} \cdot k\cdot n)+ \sum_{j= 1}^{\min(k-1, h(G))} \Oh(h(G)-1)^{j} / j^2) \cdot \\
& (\Oh(k^2 + h(G)\cdot n) + \Oh((k-j)!\cdot (h(G) - 1)^{k-j} \cdot n) + m)),
\end{align*} 
since due to~\Cref{subsets of size k}, for each~$j \geq 1$, $H$ contains at most~$\Oh((h(G)-1)^j/j^2)$ subsets of size at most~$j$.
Since~$m \in \Oh(h(G) \cdot n)$ and~$\sum_{j= 1}^k \frac{1}{j^2} \in \Oh(1)$, we obtain the stated running time.
\end{proof}

\section{Using Modular Decompositions}

Next, we provide FPT-algorithms that use modular decompositions which, roughly speaking, provide a hierarchical view of the different neighborhoods in a graph~$G$.

\subsection{Modular Decompositions}
A \emph{modular decomposition} of a graph~$G=(V,E)$ is a pair~$(\mT,\beta)$ consisting of a rooted tree~$\mT=(\mv, \mathcal{A}, x^*)$ with root~$x^*\in \mv$ and a function~$\beta$ that maps each node~$x\in \mv$ to a graph~$\beta(x)$.
If~$x$ is a leaf of~$\mT$, then~$\beta(x)$ contains a single vertex of~$V$ and for each vertex~$v\in V$, there is exactly one leaf~$\ell$ of~$\mathcal{T}$ such that the graph~$\beta(\ell)$ consists only of~$v$.
If~$x$ is not a leaf node, then the vertex set of~$\beta(x)$ is exactly the set of child nodes of~$x$ in~$\mathcal{T}$.
Moreover, let~$V_x$ denote the set of vertices of~$V$ contained in leaf nodes of the subtree rooted at~$x$.
Formally, $V_x$ is defined as~$V(\beta(\ell))$ for leaf nodes~$\ell$ and recursively defined as~$\bigcup_{y\in V(\beta(x))} V_y$ for each non-leaf node~$x$.
Moreover, we define~$G_x = (V_x,E_x) :=   G[V_x]$.
A modular decomposition has the property that for each non-leaf node~$x$ and any pair of distinct nodes~$y\in V(\beta(x))$ and~$z\in V(\beta(x))$, $y$ and~$z$ are adjacent in~$\beta(x)$ if and only if there is an edge in~$G$ between each pair of vertices of~$V_y$ and~$V_z$ and~$y$ and~$z$ are not adjacent if and only if there is no edge in~$G$ between any pair of vertices of~$V_y$ and~$V_z$.
Hence, it is impossible that there are vertex pairs~$(v_1,w_1)\in V_y \times V_z$ and~$(v_2,w_2)\in V_y \times V_z$ such that~$v_1$ is adjacent to~$w_1$ and~$v_2$ is not adjacent to~$w_2$.

We call~$\beta(x)$ the~\emph{quotient graph} of~$x$.
A quotient graph is~\emph{prime} if there is no set~$A\subsetneq V(\beta(x))$ of size at least~$2$ such that all vertices of~$A$ have the same neighborhood in~$V(\beta(x))\setminus A$.
The~\emph{width of a modular decomposition} is the size of a largest vertex set of any quotient graph and the \emph{modular-width} of a graph~$G$, denoted by~$\mw(G)$, is the minimal width of any modular decomposition of~$G$.

\subsection{Parameterization by Modular-Width}
We now provide a dynamic programming algorithm over the modular decomposition of~$G$. The nodes of the decomposition are processed in a bottom-up manner. 
The idea is to consider for a node~$x$ the possibilities of how a swap may interact with the vertex sets that are represented by the vertices~$y$ of~$\beta(x)$. 
We use the fact that any valid swap of~$G$ must also correspond in the natural way to a valid swap of~$\beta(x)$. 
More precisely, if some vertex in the set represented by~$y$ goes to the independent set, then the vertex cover must include the vertex set represented by~$z$ for all neighbors~$z$ of~$y$ in~$\beta(x)$. 
\begin{theorem}
\GLWVC can be solved in $\Oh(\mw(G)^k \cdot k \cdot (\mw(G) + k) \cdot n + m)$~time.
\end{theorem}
\begin{proof}
Let~$I=(G=(V,E),\omega,S,k,d)$ be an instance of~\GLWVC.
First, we compute a modular decomposition~$(\mathcal{T}=(\mv,\mathcal{A}, x^*),\beta)$ of minimal width in~$\Oh(n+m)$~time~\cite{MS94}. 
Note that~$\mT$ has $\Oh(n)$~nodes.
Next, we describe a dynamic program on the modular decomposition~$(\mathcal{T}, \beta)$ to solve~\GLWVC.

For each node~$x\in \mv$ in the modular decomposition, we have a dynamic programming table~$D_x$.
The table~$D_x$ has entries of type~$D_x[k']$ for~$k' \in [0, k]$.  
 Each entry~$D_x[k']$ stores the maximal improvement~$\imp_S(W)$ of a valid~$k'$-swap~$W \subseteq V_x$ for~$S\cap V_x$ in~$G_x$.

 Next, we describe how to fill the dynamic programming tables.
Let~$\ell$ be a leaf node of~$\mT$ and let~$v$ be the unique vertex of~$V(\beta(\ell))=V_\ell$. 
We fill the table~$D_\ell$ by setting~$$D_\ell[k']:=   
\begin{cases}
    0         & v\notin S \lor k' = 0\\
    \omega(v) & v\in S \land k' > 0
\end{cases}$$~for each~$k'\in[0,k]$.

To compute the entries for all remaining nodes~$x$ of~$\mT$, we use an auxiliary  table~$Q_{S_x}$.
Let~$S_x$ be an independent set in~$\beta(x)$ and let~$S_x(i)$ denote the~$i$th vertex of~$S_x$ according to some arbitrary but fixed ordering with~$i\in[1,|S_x|]$.
Moreover, let~$V^{\leq i}_x = \bigcup^{i}_{j=1} V_{S_x(j)}$.
Recall that~$S_x(j)$ is both a vertex of~$\beta(x)$ and a child node of~$x$ in~$\mT$ and that~$V_{S_x(j)}$ denotes the set of all vertices of~$G$ that are contained in the subtree of~$\mT$ rooted at~$S_x(j)$.
The dynamic programming table~$Q_{S_x}[i,k']$ has entries for~$i\in [1, |S_x|]$ and~$k'\in[0,k]$ and stores the maximal improvement of a valid~$k'$-swap~$W$ for~$S \cap V^{\leq i}_x$ in~$G[V^{\leq i}_x]$, such that for each~$j\in[1, i]$, at least one vertex of~$S \cap V_{S_x(j)}$ is contained in~$W$.
We set~$$Q_{S_x}[i, k'] :=   \begin{cases}
- \infty & k' < i,\\
D_{S_x(1)}[k']     & i = 1,\text{and}\\
\max_{1 \leq k'' \leq k'} D_{S_x(i)}[k''] + Q_{S_x}[i-1, k'-k''] & \text{otherwise}.
\end{cases}$$ 
Since we are looking for a~$k'$-swap~$W$ that contains for each~$j\in[1,i]$ at least one vertex of~$S\cap V_{S_x(j)}$, the value of the table is set to~$-\infty$ if~$k' < i$, since there is no such swap of size at most~$k' < i$. 
If~$k' \geq i$, then the value of the table is determined by finding the best way to swap~$k''$ vertices of~$V_{S_x(i)}$ and~$k'-k''$ vertices of~$V^{\leq i-1}_x$.

The entries for~$D_x$ can then be computed as follows:
$$D_x[k'] :=   \max_{\stackrel{S_x \subseteq V(\beta(x))}{\stackrel{|W^*| \leq k'}{S_x~\text{is independent}}}} Q_{S_x}[1, k' - |W^*|] - \omega(W^*)$$
where~$W^* :=   \bigcup_{y\in N_x(S_x)}(V_y \setminus S)$.
Since we are looking for a swap~$W$ where for each~$y\in S_x$, at least one vertex of~$S\cap V_y$ is contained in~$W$ and thus leaves the vertex cover, each vertex of~$W^*$ has to be added to obtain a vertex cover. 

The maximal improvement of any valid~$k$-swap for~$S$ in~$G$ can be found in~$D_{x^*}[k]$, where~$x^*$ is the root of the modular decomposition.
Moreover, the corresponding~$k$-swap can be found via traceback.

Next, we analyze the running time.
For each non-leaf node~$x$, and each independent set~$S_x$ of size at most~$k$ in~$\beta(x)$, there are~$\Oh(k^2)$ table entries in~$Q_{S_x}$ and each of these entries can be computed in $\Oh(k)$~time.
Recall that for a set of size~$x$, $\leqBin{x}{k}$ denote the number of different subsets of size at most~$k$.

Since each quotient graph has~$\Oh(\leqBin{\mw(G)}{k})$ many independent sets of size at most~$k$, all entries of all tables~$Q_{S_x}$ can be computed in $\Oh(\leqBin{\mw(G)}{k}\cdot k^3\cdot n)$~time, since the modular decomposition has $\Oh(n)$~quotient graphs.
For each node~$x$, there are~$\Oh(k)$ table entries in~$D_x$.
We will show that we can compute each of them in $\Oh(k^2\cdot (\mw(G) + k))$~time.
To this end, we precompute for each node~$x$ the size~$|V_x \setminus S|$ and the weight~$\omega(V_x \setminus S)$ to compute~$|W^*|$ and~$\omega(W^*)$ in $\Oh(k)$~time afterwards.
Since for all non-leaf nodes~$x$, $V_x \setminus S = \bigcup_{y\in V(\beta(x))} (V_y \setminus S)$, we can compute~$|V_x \setminus S|$ as~$\sum_{y\in V(\beta(x))}|V_y \setminus S|$ and~$\omega(V_x \setminus S)$ as~$\sum_{y\in V(\beta(x))}\omega(V_y \setminus S)$.
This can be done in $\Oh(\mw(G) \cdot n)$~time since the modular decomposition has $\Oh(n)$~quotient graphs.
Hence, for an independent set~$S_x$ of size at most~$k$, we can compute~$|W^*|$ and~$\omega(W^*)$ in $\Oh(k)$~time.
Since we can enumerate all subsets~$S_x$ of size at most~$k$ of~$V(\beta(x))$ in~$\Oh(\leqBin{\mw(G)}{k})$~time and check in~$\Oh(\mw(G) \cdot k)$~time if~$S_x$ is independent in~$\beta(x)$, we can compute~$D_x[k']$ in $\Oh(\leqBin{\mw(G)}{k}\cdot k^2 \cdot (k+\mw(G)))$~time.
Consequently, we can compute all entries of the dynamic programming tables in~$\Oh(\leqBin{\mw(G)}{k}\cdot k^3 \cdot (k+\mw(G)) \cdot n + m)$~time, which is $\Oh(\mw(G)^k \cdot k \cdot (\mw(G) + k) \cdot n + m)$~time due to~\Cref{subsets of size k}.

Since the value of~$Q_{S_x}[1, k']$ is only evaluated once during the whole computation of this dynamic programming algorithm, we can remove the table~$Q_{S_x}$ after evaluating~$Q_{S_x}[1, k']$ for each~$k'$.
Consequently, this algorithm also only uses polynomial space.
\end{proof}

Note that with a slight modification of the dynamic programming algorithm, we can improve the running time for the special case of~\GLVC.
Recall that~$I$ is a yes-instance of~\GLVC if and only if there is a \goodswap~$W$ for~$I$ with~$|W\cap S| \leq \frac{k+d}{2}$.
Hence for~\GLVC, it is sufficient in the computation of~$D_x[k']$ to only check for independent sets~$S_x$ in~$\beta(x)$ of size at most~$\min(k', \frac{k+d}{2})$.
Thus, we obtain the following.

\begin{corollary}\label{cor:modwidthgap}
\GLVC can be solved in $\Oh(\mw(G)^{\frac{k+d}{2}} \cdot k \cdot (\mw(G) + k) \cdot n + m)$~time.
\end{corollary}

\subsection{Maximum Degree of the Modular Decomposition}
We also obtain an FPT-algorithm for~\GLVC{} for a new parameter that is upper-bounded by the maximum degree~$\Delta(G)$ and by the modular-width~$\mw(G)$. 
We call this parameter the \emph{maximum modular degree} and it is defined by taking the maximum degree over all quotient graphs of a modular decomposition of minimum width.
\begin{definition}
  Let~$(\mT=(\mv, \mathcal{A}, x^*),\beta)$ be a modular decomposition of a graph~$G$.
  Then the \emph{maximum modular degree} of~$(\mT,\beta)$ is $\Delta_{\md}(\mT,\beta):=  \max_{x\in \mv} \Delta(\beta(x))$.
  Moreover, the \emph{maximum modular degree}~$\Delta_{\md}(G)$ of~$G$ is the maximum modular degree of a modular decomposition~$(\mT',\beta')$ of~$G$ that minimizes~$\Delta_{\md}(\mT',\beta')$. 
\end{definition}
Since~$\mw(G)$ is the largest vertex count of any quotient graph, we have~$\Delta_{\md}(G) < \mw(G)$.
Moreover, for each node~$x\in V$, the graph~$\beta(x)$ is isomorphic to an induced subgraph of~$G$, which implies that~$\Delta_{\md}(G)\le \Delta(G)$.

Before describing the algorithm, we show that~$\Delta_{\md}(G)$ can be easily computed by inspecting a minimum-width modular decomposition. 
\begin{proposition}
Let~$(\mT=(\mv, \mathcal{A}, x^*),\beta)$ be a modular decomposition of a graph~$G$ of minimum width where the quotient graph~$\beta(x)$ is prime for each~$x\in \mv$.
Then, $\Delta_{\md}(\mT,\beta) = \Delta_{\md}(G)$.
\end{proposition}
\begin{proof}
Let~$x\in \mv$ be a node of~$\mT$ such that the degree of~$\beta(x)$ is exactly~$\Delta_{\md}(\mT,\beta)$. Recall that there is a vertex set~$S\subseteq V(G)$ such that~$G[S]$ is isomorphic to~$\beta(x)$. 

Let~$(\mT'=(\mv', \mathcal{A}', y^*),\beta')$ be a modular decomposition of~$G$.
We show that~$(\mT',\beta')$ contains a quotient graph~$\beta'(x')$ which has an induced subgraph isomorphic to~$\beta(x)$. This implies~$\Delta_{\md}(\mT',\beta') \geq \Delta_{\md}(\mT,\beta)$.
Let~$\beta'(x')$ be the quotient graph of~$(\mT',\beta')$ such that~$S \subseteq V_{x'}$ and that for each descendant~$x''$ of~$x'$ in~$\mT'$, $S \not\subseteq V_{x''}$.
Moreover, let~$Z' :=   \{z \in V(\beta'(x'))\mid S \cap V_{z} \neq \emptyset\}$.
Note that~$|Z'| \leq |S|$ and by the above, $|Z'|\geq 2$.

Assume towards a contradiction that there is some~$z\in Z'$ such that~$S_z :=   V_{z}\cap S$ contains at least two vertices.
Since~$\beta(x)$ and thus~$G[S]$ is prime, there are vertices~$v_1$ and~$v_2$ in~$S_z\subseteq S$ and a vertex~$w\in S \setminus S_z$ such that~$w$ is adjacent to exactly one of~$v_1$ and~$v_2$.
Let~$w'$ be the vertex of~$\beta'(x')$ such that~$w\in V_{w'}$.
If $z$~and~$w'$ are adjacent~$\beta'(x')$, then each vertex of~$V_{w'}$ is adjacent to each vertex of~$V_{z}$ in~$G$. Similarly, $z$~and~$w'$ are adjacent~$\beta'(x')$, then no vertex of~$V_{w'}$ is adjacent to any vertex of~$V_{z}$ in~$G$. In both cases, we arrive at a contradiction the choice of~$w$ as being adjacent to exactly one of~$v_1$ and~$v_2$.

Hence, ~$|V_z \cap S| = 1$ for all~$z\in Z'$ and thus~$|Z'|=|S|$.
For each~$v\in S$, let $v'$ be the vertex of~$Z'$ with~$v\in V_{v'}$.
Since~$(\mT',\beta')$ is a modular decomposition, for each two distinct vertices~$u,v\in S$, $\{u,v\}$ is an edge in~$G[S]$ if and only if~$\{u',v'\}$ is an edge in~$\beta'(x')$.  
Hence, the subgraph of $\beta'(x')$ induced by~$Z'$ is isomorphic to~$G[S]$.
As a consequence, $\Delta_{\md}(\mT',\beta') \geq \Delta_{\md}(\mT,\beta)$ which then implies~$\Delta_{\md}(G) = \Delta_{\md}(\mT,\beta)$.
\end{proof}

In the following we present an algorithm that solves \GLVC in time~$(\Delta_{\md}(G)\cdot k)^{\Oh(k)}\cdot n^{\Oh(1)}$.
As in the dynamic program for the modular-width~$\mw(G)$, the overall strategy is to compute for each node~$x$ of the decomposition and each~$k'\in [0,k]$ the value~$D_x[k']$ which stores the maximal improvement~$\imp_S(W)$ of a valid~$k'$-swap~$W \subseteq V_x$ for~$S\cap V_x$ in~$G_x$. 
The difference lies in how this value is computed for a given node~$x$ of the modular decomposition. Instead of considering all interactions of $k'$-swaps with~$\beta(x)$ by brute force, we enumerate the connected $\widetilde{k}$-swaps of~$\beta(x)$ that achieve the best improvement for~$\widetilde{k}\le k'$. 
Here again, we face the difficulty that the best $k'$-swap may not be connected. 
To overcome this difficulty, we use again a branching strategy: if no best connected $\widetilde{k}$-swap for any~$\widetilde{k}\le k'$ is part of an optimal~$k'$-swap, then one of its vertices or some neighbor of its vertices must be part of some optimal~$k'$-swap. 

To give a clean description of the algorithm, we introduce a generalization of \LVClong{}. 
Intuitively, in this problem we aim to find a swap~$W$ for a given vertex cover~$S$ such that the cost of the new vertices in the vertex cover~$W\setminus S$ plus the best cost improvement that we can get by distributing the remaining swap budget among the vertices that are not in the new vertex cover~$W\oplus S$ is small. 
Essentially, this new problem is a combination of \LVC and a \textsc{Knapsack} problem on the independent set vertices.

\taskprob{\SWLVC}{An undirected graph~$G=(V,E)$, with a vertex cover~$S$, an integer~$k$, an \emph{internal swap reward function}~$\gamma\colon V\times [0,k]\to \mathbb{N}$ and \emph{cover-cost} function~$\delta\colon V\setminus S\to \mathbb{N}$.}
{Compute a valid swap~$W\subseteq V$ for~$S$ in~$G$ and an \emph{internal swap number}~$c(v)$ for each vertex~$v\in V\setminus (W\oplus S)$ such that

  \begin{itemize}
  \item $c(v) \geq 1$ for each $v\in W\cap S$,
  \item  $\sum_{v\in V\setminus (W\oplus S)} c(v)  + \sum_{v\in W\setminus S} \delta(v)\le k$, and 
  \item $\imp^S(W,c):=  \sum_{v\in V\setminus (W\oplus S)} \gamma(v,c(v))  - \sum_{v\in W\setminus S} \delta(v)$ is maximum.
  \end{itemize}
} 
For a given instance~$I$ of~\SWLVC, we let~$\val(I)$ denote the
improvement~$\imp^S(W,c)$ obtained by an optimal swap.  

The intuition behind this problem formulation and its application to modular decompositions is as follows. When we want to compute the values of~$D_x$ for a node~$x$ of a modular decomposition, we consider how a swap interacts with~$\beta(x)$. There are some vertices~$v$ of~$\beta(x)$ that are added to the vertex cover by the swap, in the sense that \emph{all} vertices of~$V_v$ are added to the vertex cover. The value of~$\delta(v)$ is exactly the number of vertices of~$V_v$ that are not yet in the vertex cover. For the vertices~$v$ of~$\beta(x)$ such that~$V_v$ is not fully contained in the vertex cover, one may still swap a certain number of vertices of~$V_v$ in order to further decrease the vertex cover weight. To obtain a maximal improvement, we may distribute the remaining swap budget on the disjoint vertex sets~$V_v$,~$v\in \beta(x)$. This distribution is modelled by the internal swap numbers~$c$. The function~$\gamma$ represents how much improvement we get for a certain number of internal swaps.

With this intuition in mind, we reduce the computation of the table values of~$D_x[k]$ for a \GLVC instance with graph~$G'$ and vertex cover~$S'$ to \SWLVC{} as follows:  
The graph~$G$ is simply~$\beta(x)$. 
For each vertex~$y$ of~$\beta(x)$ and each~$k'\in [0,k]$, set~$\gamma(y,k'):=  D_y[k']$. 
The vertex cover~$S$ of~$G$ consists of all vertices~$y\in \beta(x)$ with~$V_y\subseteq S'$. 
Finally, for all other vertices~$y$ of~$\beta(x)$ set~$\delta(y):=|V_y\setminus S'|$; these are the costs of adding the remaining vertices of~$V_y$ to the vertex cover. 

\begin{proposition}
  Let~$D_x[k]$ be the maximal improvement of a valid $k$-swap for $S'\cap V'_x$ in~$G'_x$ and let~$I$ be the \SWLVC{} instance constructed as described above. Then,~$D_x[k]=\val(I)$.
\end{proposition}

We now describe how to solve \SWLVC{}. Observe that every optimal swap for
a given~$k$ may be disconnected since we may need to combine different connected
swaps with a positive improvement value.
 
The first step of the algorithm is to compute for each~$k'\in [1,k]$ some connected swap~$W'$ that gives a maximum improvement among all connected swaps. 
To this end, the algorithm enumerates for all~$\widetilde{k}\in [1,k]$ all connected valid size-$\widetilde{k}$~swaps~$\widetilde{W}$ of~$G$ and~$S$. 
For each such valid swap~$\widetilde{W}$, we compute the internal swap numbers~$c(v)$ for the vertices~$v$ in the independent set such that the improvement is maximized without exceeding the total budget for the connected swap. 
Let~$\{v_1,\ldots,v_t\}$ denote the vertices of~$V\setminus (\widetilde{W}\oplus S)$. We compute the optimal values of~$c$ for~$\{v_1,\ldots,v_t\}$ by dynamic programming: 
Compute table entries~$T[i,k']$ that store the maximum value of~$\sum_{j\in [1,i]} \gamma(v_j,c({v_j}))$ such that $\sum_{j\in [1,i]} c({v_j})= k'$.
Observe that this can be done in~$\Oh(k\cdot n)$~time using the standard knapsack dynamic programming algorithm. 
After computing~$T$, for each~$k'\in [\widetilde{k},k]$ the optimal improvement that can be obtained for the swap~$\widetilde{W}$ by investing a swap budget of~$k'\ge \widetilde{k}$ is

$$T[t,k'- \sum_{v\in \widetilde{W}\setminus S} \delta(v)] - \sum_{v\in \widetilde{W}\setminus S} \delta(v).$$

This gives a set of improvement values~$\imp^S_{k'}(\widetilde{W})$, $k'\in [\widetilde{k},k]$ for
the swap~$\widetilde{W}$. 
After computing these values, check for each~$k'$ whether
$\widetilde{W}$ gives a better improvement value~$\imp^S_{k'}$ than previous connected
swaps. If this is the case, then store~$\widetilde{W}$ as the currently best connected swap for the swap budget~$k'$.

After all connected swaps of size at most~$k$ have been considered, we have computed for each~$i\in [1,k]$ some connected swap~$W_{i}$ that maximizes~$\imp^S_{i}$ among all connected swaps with swap budget~$i$. 
Now we apply a branching as in the algorithm for \GLVC: either swap at least one vertex of one of these connected swaps~$W_i$ or a neighbor of some vertex of~$W_i$.   

To formulate the branching, we again use swap-instances. 
When we send a vertex~$w$ in the vertex cover~$S$ to the independent set, we will send all its neighbors permanently to the vertex cover which is essentially the same as removing them from~$G$ and paying the~$\delta$-values for those that are not yet in~$S$.   
When a vertex~$w$ that is in the independent set performs an internal swap, then this means that~$w$ stays permanently in the independent set. In turn, this implies that all its neighbors stay permanently in the vertex cover and we can simply remove them from~$G$.
In both cases, we pay exactly one internal swap for~$w$, deferring any additional internal swaps to the constructed swap-instance. Accordingly, in both cases we decrease~$k$ by~one and shift the~$\gamma$-values of~$w$ down by one, that is, $\gamma'(w,k') :=   \gamma(w,k'+1)$.
The formal definition of the swap-instances reads as follows.
\begin{definition}
  Let~$I=(G=(V,E),S,k,\gamma,\delta)$ be an instance of~\SWLVC and let~$w$ be a vertex of~$G$.
  If~$w\in S$, then

  $$\swap(I,w):=  
    (G-N(w),S\setminus \{w\},k':=  k-1-\sum_{v\in N(w)\setminus S}\delta(v),\gamma',\delta').$$
    Otherwise, if~$w\notin S$, then,    
    $$
    \swap(I,w):=  (G-\{w\},S,k':=  k-1,\gamma',\delta') .
    $$
    Herein, (1)~$\gamma'$ and~$\gamma$ agree on~$(V(G)\setminus N[w])\times [0,k']$,  (2) $\gamma'(w,\widetilde{k}):=   \gamma(w, \widetilde{k}+1)$ for each~$\widetilde{k}\in [0,k']$, and (3) $\delta'$ and~$\delta$ agree on~$V(G) \setminus N[w]$ and~$\delta'(w) := \infty$. 
\end{definition}

Before we state the branching rule, we need to specify the optimal values for the base case, when~$W$ is empty and the swap budget~$k$ (which may be zero) is distributed among the vertices that are in~$V\setminus S$. In that case, we need to compute numbers~$c(v)$ for each~$v\in V\setminus S$ such that $$\sum_{v\in V\setminus S} c(v)  \le k$$ and~$$\sum_{v\in V\setminus S} \gamma(v,c(v))$$ is maximum. We denote this value by~$\imp^S_k(\emptyset)$. 

We may now state the branching rule.
\begin{brrule}\label{rule:md-swaps}
Let~$I=(G=(V,E),S,k,\gamma,\delta)$ be an instance of~\SWLVC and, for each~$i\in[1,k]$, let~$W_i$ denote a valid connected swap of size~$i$ for~$S$ maximizing~$\imp^S_i$ among all valid connected swaps of size~$i$ for~$S$.
Return the maximum of~$\imp^S_k(\emptyset)$ and~$$\mathop{\max}_{w\in N[W_i], i\le k}\quad  \branch(I,w),$$
where $\branch(I,w)= \val(\swap(I,w))-\sum_{v\in N(w)\setminus S}\delta(v).$
\end{brrule}

The correctness proof for this branching rule is very similar to the proof of the correctness of~\Cref{rule:swaps with weight}.
Let~$W$ be some valid swap together with internal swap numbers~$c_v$ for each~$v\in V\setminus(W\oplus S)$ such that~$\imp_k^S(W)$ is maximum.

Suppose that~$W\neq \emptyset$.
Then, there is a connected component~$C$ in~$G[W]$ that minimizes~$\ell :=   \sum_{v\in V\setminus(C\oplus S)} c_v + \sum_{v\in C \setminus S} \delta(v)$.
Since~$\imp^S_\ell :=   \imp^S_\ell(W_\ell)$ is maximal, $\imp^S_\ell(W_\ell) \geq \imp^S_\ell(C)$.
Hence, one can replace the swap~$C$ in~$W$ by~$W_i$.
The only reason why this does not result in a valid swap, is because~$W$ contains at least one vertex of~$N[W_i]$ and performs an internal swap of size at least one, such internal swaps are considered in the swap-instances.
Since~\Cref{rule:md-swaps} builds the maximum over~$\imp^S_k(\emptyset)$ and all swap-instances of vertices contained in~$\bigcup_{i= 1}^k N[W_i]$, the branching is correct.

\begin{theorem}
  \GLVC{} can be solved in~$\left(\Delta_{\md}(G)(k^2+2)\right)^k\cdot n^{\Oh(1)}$ time.
\end{theorem}
\begin{proof}
  The algorithm uses dynamic programming over the modular decomposition of~$G$. More precisely, it computes~$D_x[k']$ for every node~$x$ of the modular decomposition and every~$k'\in [0,k]$. The instance is a yes-instance if and only if~$D_{x^*}[k]\ge d$ where~$x^*$ is the root of the modular decomposition.

  To bound the running time, it is sufficient to show that~$D_x[k']$ can be computed in ~$(\Delta_{\md}(G)\cdot k^2)^k\cdot n^{\Oh(1)}$ time for non-leaf nodes~$x$ of the decomposition. The reduction to the \SWLVC~instance can be clearly computed in polynomial time. 
  \SWLVC{} is solved by a search tree algorithm that applies Branching Rule~\ref{rule:md-swaps} as long as~$k\ge 0$. 
  In every instance that is created by the branching rule, the parameter~$k$ is decreased by at least~1. 
  Thus, the search tree has depth at most~$k$. 
  The number of instances that are created by the rule for~$k>1$ is at most~$\Delta_{\md}(G)\cdot k^2$: There are at most~$k$ sets~$W_j$ and each set~$W_j$ contains at most~$k$ vertices. 
  Thus, $N[W_j]$ contains at most~$\Delta_{\md}(G) + 1<\Delta_{\md}(G)\cdot k$~vertices for~$j=1$ and at most~$(\Delta_{\md}(G)-1)\cdot k + k= \Delta_{\md}(G)\cdot k$ vertices for~$j>1$. 
  For~$k\le 1$, the number of recursively created instances is~$n^{\Oh(1)}$. 
  Thus, the search tree has size~$(\Delta_{\md}(G)k^2)^k\cdot n^{\Oh(1)}$. 
  The only non-polynomial running time factor incurred at each search tree node is due to the enumeration of connected swaps which can be done in~$(2\Delta_{\md}(G))^{k}\cdot n^{\Oh(1)}$~time due to Proposition~\ref{running time rule small swaps weighted}. The overall running time follows.   
\end{proof}

\section{Using Split Decompositions}

Finally, we provide an FPT-algorithm for the splitwidth of~$G$, another parameter that is upper-bounded by the modular-width of~$G$.
First, we define split decompositions.

A~\emph{split} of a graph~$G=(V,E)$ is a partition~$(V_1, V_2)$ of~$V$ with~$|V_1| \geq 2$ and~$|V_2|\geq 2$ such that each vertex in~$V_1$ with at least one neighbor in~$V_2$ has the same neighborhood in~$V_2$.
In other words, there are sets~$V_1'\subseteq V_1$ and~$V_2'\subseteq V_2$ such that~$N(v)\cap V_2 = V'_2$ for each~$v\in V'_1$ and~$N(w)\cap V_2 = \emptyset$ for each~$w\in V_1\setminus V'_1$.
If there is no split for~$G$, we call~$G$~\emph{prime}.
Let~$(V_1,V_2)$ be a split of~$G$.
A~\emph{simple decomposition of~$G$ with respect to~$(V_1,V_2)$} consists of two graphs~$G_1=(V_1,E_1)$ and~$G_2=(V_2,E_2)$ where for each~$i\in \{1,2\}$
\begin{itemize}
\item $W_i = V_i \cup \{x\}$ for some vertex~$x$ which is not contained in~$V$,
\item $G_i[V_i] = G[V_i]$,
\item and~$x$ is adjacent to exactly the vertices of~$V'_i$ in~$G_i$.
\end{itemize}
The vertex~$x$ is called a~\emph{marker} vertex.
Conversely, two graphs~$G_1=(V_1,E_1)$ and~$G_2=(V_2,E_2)$ with~$V_1\cap V_2 = \{x\}$ can be \emph{composed} into a graph~$G :=  (V,E)$ as follows: 
The graph~$G$ is the union of~$G_1$ and~$G_2$ without the marker vertex~$x$ plus the edges between each neighbor of~$x$ in~$G_1$ and each neighbor of~$x$ in~$G_2$.
Formally, $V :=   (V_1 \cup V_2) \setminus \{x\}$ and~$E:=  E_{G_1}(V_1\setminus \{x\}) \cup E_{G_2}(V_2\setminus \{x\}) \cup \{\{v_1, v_2\}\mid v_1 \in N_{G_1}(x), v_2 \in N_{G_2}(x)\}$.
Note that~$(V_1\setminus \{x\}, V_2\setminus\{x\})$ is a split for~$G$ and~$G_1$ and~$G_2$ are a simple decomposition of~$G$ with respect to~$(V_1\setminus \{x\}, V_2\setminus\{x\})$.
A~\emph{split decomposition}~$(\mathcal{T},\beta)$ of a graph~$G$ consists of an undirected tree~$\mathcal{T}=(\mv, \mathcal{E})$ and a function~$\beta$ that maps each of the nodes~$x$ of~$\mv$ to a prime graph~$\beta(x)$ such that
\begin{itemize}
\item $|V(\beta(x)) \cap V(\beta(y))| = 1$ if~$x$ and~$y$ are adjacent in~$\mT$,
  \item $V(\beta(x)) \cap V(\beta(y)) = \emptyset$ if~$x$ and~$y$ are not adjacent in~$\mT$, and 
\item $G$ is equivalent to the graph obtained from recursively composing the graphs of all adjacent node pairs in~$\mT$.
\end{itemize} 
The~\emph{width of a split decomposition} is the size of the largest vertex set of any prime graph and the \emph{splitwidth} of a graph~$G$ is the minimal width of any split decomposition of~$G$ denoted by~$\sw(G)$.
A split decomposition of minimal width for a graph~$G$ can be computed in linear time~\cite{D00}.

For the following dynamic programming algorithm, we use a so-called nice split decomposition, which has some helpful properties.
A~\emph{nice split decomposition} of~$G$ consists of a rooted tree~$\mathcal{T}=(\mathcal{N}, \mathcal{A},x^*)$ and a mapping~$\beta$ and can be obtained from the split decomposition~$(\mT', \beta')$ as follows: 
First, in each prime graph~$\beta'(x)$, replace each non-marker vertex~$v$ by a marker vertex~$y_v$ and add a new neighbor~$v$ to~$x$ in~$\mT'$ where~$\beta'(v)$ consists of the single edge~$\{v, y_v\}$.
Note that afterwards, for each vertex~$v\in V$ there is only one node~$\ell$ of~$\mT'$ such that~$v\in \beta'(\ell)$ and this node is a leaf node. 
Second, we take an arbitrary non-leaf node of~$\mT'$ as the root~$x^*$ and orient all edges so that we obtain an out-tree.
For simplicity, we rename each non-root node~$x$ of~$\mT$ to~$a$, where~$a$ is the marker vertex contained in the prime graph of~$x$ and the prime graph of the unique parent of~$x$ in~$\mT$. To~$\beta(x^*)$, the prime graph of the root~$x^*$, we add an auxiliary marker vertex~$x^*$ with no neighbors.
After these replacements, for each non-leaf node~$x$ of~$\mT$, the vertex set of the prime graph~$\beta(x)$ is exactly the set of child nodes of~$x$ in~$\mT$ together with the marker vertex~$x$ itself.
Note that a nice split decomposition can be computed from a given split decomposition in linear time.

Similar to modular decomposition and tree decomposition, we define for a node~$x$ of the nice split decomposition~$V_x$ as the vertices of~$V$ in any prime graph of the subtree of~$\mT$ rooted at~$x$.
This can again be recursively defined as~$V_\ell :=   V(\beta(\ell))\cap V$ for each leaf node~$\ell$ and as~$V_x :=   \bigcup_{y\in V(\beta(x))}V_y$.
Moreover, we define the~\emph{border}~$\bor(x)$ of~$x$ as the set of vertices of~$V_x$ which are adjacent to the marker vertex~$x$ after recursively composing all adjacent prime graphs of the subtree of~$\mT$ rooted at~$x$.
This set can also be recursively defined as~$\bor(\ell) :=   V_\ell$ for each leaf node~$\ell$ and as~$\bor(x) :=   \bigcup_{y\in N_x(x)}\bor(y)$.
For the root~$x^*$ of~$\mT$, the auxiliary vertex~$x^*$ has no neighbors in~$\beta(x^*)$, so $\bor(x^*) = \emptyset$.

The main difference between the algorithm for modular-width and splitwidth is how to treat border vertices. 
The idea behind border vertices is that for a non-leaf node~$x$ and adjacent marker vertices~$y$ and~$z$ in~$\beta(x)$ each vertex~$v_y$ of~$\bor(y)$ is adjacent to each vertex~$v_z$ of~$\bor(z)$ and there is no edge between any other pair of vertices of~$V_y$ and~$V_z$.
Consequently, if a vertex of~$\bor(y)\cap S$ is contained in some valid swap~$W$ for~$S$ in~$G$, then no vertex of~$\bor(z)\cap S$ is contained in~$W$ and all vertices of~$\bor(z) \setminus S$ are contained in~$W$.
Hence, in a prime graph~$\beta(x)$, the set of marker vertices~$y$ where at least one vertex of~$\bor(y) \cap S$ is swapped, forms an independent set in~$\beta(x)$.
The following dynamic programming mainly relies on this observation and distinguishes between all possible intersections of a swap~$W$ with~$\bor(x)$.

 \begin{theorem}
\GLWVC can be solved in $\Oh(\sw(G)^{k+1} \cdot k \cdot n + m)$~time, where~$\sw(G)$ is the splitwidth of~$G$.
 \end{theorem}
 \begin{proof}
 Let~$I=(G=(V,E),\omega,S,k,d)$ be an instance of~\GLWVC.
First, we compute a nice split decomposition~$(\mT'=(\mv,\mathcal{E}),\beta')$ of minimal width in~$\Oh(n+m)$~time~\cite{D00}. 
 
Next, for each node~$x$ of~$\mT$, we store three dynamic programming tables~$D^+_x[k']$, $D^-_x[k']$, and~$D^o_x[k']$ each with entries for~$k'\in[1,k]$.
Each entry stores the maximal improvement of a valid~$k'$-swap~$W$ for~$S\cap V_x$ in~$G[V_x]$ such that
\begin{itemize}
\item in case of~$D^+_x[k']$, vertices of~$\bor(x)\cap S$ can be contained in~$W$,
\item in case of~$D^-_x[k']$, all vertices of~$\bor(x)\setminus S$ are contained in~$W$ and no vertex of~$\bor(x)\cap S$ is contained in~$W$, and
\item in case of~$D^o_x[k']$, no vertex of~$\bor(x)\cap S$ is contained in~$W$.
\end{itemize}

Let~$\ell$ be a leaf of~$\mT$ and let~$v$ be the unique vertex of~$V$ in~$V_\ell$.
Recall that the prime graph~$\beta(\ell)$ consists of a single edge between the marker vertex~$\ell$ and the vertex~$v$.
We fill the dynamic programming tables for~$\ell$ as follows. 

$$ D^+_\ell[k'] :=
\begin{cases} 0 & k' = 0 \lor v\notin S,\\
  \omega(v)&\text{otherwise}.
\end{cases}
$$

$$D^-_\ell[k'] :=
\begin{cases}
  0 & v\in S,\\
  -\infty & v\notin S \land k' = 0,\\
  -\omega(v)&\text{otherwise}.
\end{cases}
$$

$$D^o_\ell[k'] :=   0.$$

For all non-leaf nodes~$x$, we fill the dynamic programming tables as
follows:
 $$D^{+}_x[k'] :=   \max_{\stackrel{S_x \subseteq V(\beta(x)) \setminus \{x\}}{\stackrel{|S_x| \leq k'}{S_x~\text{is independent}}}}  Q_{S_x}[1, k']
 $$
 
 $$D^{-}_x[k'] :=   \max_{\stackrel{S_x \subseteq V(\beta(x))\setminus N_x[x]}{\stackrel{|S_x| \leq k'}{S_x~\text{is independent}}}}  Q_{S_x \cup \{x\}}[1, k']
 $$
 
 $$D^{o}_x[k'] :=   \max_{\stackrel{S_x \subseteq V(\beta(x))\setminus N_x[x]}{\stackrel{|S_x| \leq k'}{S_x~\text{is independent}}}}  Q_{S_x}[1, k']
 $$
 
 In all three cases, $S_x$ denotes the set of marker vertices~$y$ of~$V(\beta(x))$ for which vertices of~$\bor(y)$ can be contained in~$W$.
 As mentioned above, this set~$S_x$ is an independent set as otherwise, the swap~$W$ is not valid.
 The auxiliary dynamic programming table~$Q_{S_x}$ is only used to prevent an exponential factor in the running time while finding an optimal distribution of the budget~$k'$ into the vertex sets~$V_y$ for~$y\in V(\beta(x))$.

 For each node~$x$, fix an arbitrary ordering of the vertices of~$V(\beta(x))$ where the marker vertex~$x$ appears in the last position of this ordering.
 Moreover, let~$V(\beta(x))(i)$ denote the~$i$th vertex of~$V(\beta(x))$ according to this ordering.
 The dynamic programming table~$Q_{S_x}[i,k']$ has entries for~$i\in[1, |V(\beta(x))|]$ and~$k'\in[0,k]$ and stores the maximal improvement of a valid~$k'$-swap~$W$ for~$S\cap \bigcup_{j= i}^{|V(\beta(x))| - 1}V(\beta(x))(j)$ in~$G[\bigcup_{j= i}^{|V(\beta(x))| - 1}V_{\beta(x)(j)}]$ such that vertices of~$\bor(y)\cap S$ can be contained in~$W$.
 
 $$Q_{S_x}[|V(\beta(x))|, k'] :=   0 $$
 $$Q_{S_x}[i,k'] :=    \max_{1\leq k'' \leq k'} D_{S_x}[i+1, k' - k''] + 
 \begin{cases}
    D^+_y[k'']& y \in S_x \setminus \{x\},\\
    D^-_y[k'']& y \in N_x(S_x) \setminus \{x\},\\
    D^o_y[k'']&\text{otherwise}.
 \end{cases}
 $$
 where~$y$ is the~$i$th vertex of~$V(\beta(x))$.

If~$y\in S_x$, vertices of~$\bor(y)\cap S$ can be contained in~$W$, and we evaluate~$D^+_y$.
Moreover, if~$y\in N_x(S_x)$, vertices of~$\bor(y)\setminus S$ are contained in~$W$, and we evaluate~$D^-_y$.
In all remaining cases, we have to evaluate~$D^o_y$.

Recall that in the root~$x^*$, the marker vertex~$x^*$ is an isolated vertex. 
Hence for all~$k'\in[1,k]$, $D^+_{x^*}[k'] = D^-_{x^*}[k'] = D^o_{x^*}[k']$.
The maximal improvement of a valid~$k$-swap for~$I$ can be found in~$D^+_{x^*}[k']$ and the corresponding swap can be found via traceback.

Next, we analyze the running time.
The split decomposition has $\Oh(n)$~nodes.
For each such node~$x$, we have three dynamic programming tables~$D^+_x[k'], D^-_x[k'],$ and~$ D^o_x[k']$ each with~$\Oh(k)$~entries.
We show that we can compute each such entry in $\Oh(\leqBin{\sw(G)}{k}\cdot \sw(G)\cdot k^2)$~time.
Since~$|V(\beta(x)) \setminus \{x\}| \leq \sw(G) $, we can enumerate all subsets~$S_x$ of size at most~$k$ of~$V(\beta(x))\setminus \{x\}$ in $\Oh(\leqBin{\sw(G)}{k})$~time.
For each such set~$S_x$, we can check in $\Oh(\sw(G) \cdot k)$~time if~$S_x$ is an independent set in~$\beta(x)$, since~$S_x$ has size at most~$k$ and each vertex of~$\beta(x)$ has at most~$\sw(G)$ neighbors in~$\beta(x)$.
In the same running time, we can compute~$N_x(S_x)$.
If~$S_x$ is an independent set, we compute~$Q_{S_x}[1,k']$.
Since~$N_x(S_x)$ is already computed, $Q_{S_x}$ has $\Oh(\sw(G) \cdot k)$~entries, and each such entry can be computed in $\Oh(k)$~time, $Q_{S_x}[1,k']$ can be computed in $\Oh(\sw(G) \cdot k^2)$~time.
Consequently, the value of~$D^+_{x^*}[k']$ can be computed in $\Oh(\leqBin{\sw(G)}{k}\cdot \sw(G)\cdot k^3 \cdot n + m)$~time.
Due to~\Cref{subsets of size k} this is $\Oh(\sw(G)^{k+1}\cdot k \cdot n + m)$~time.
 
 Note that for each~$k'$, the value of~$Q_{S_x}[1,k']$ is used at most two times during the whole computation of all entries of all dynamic programming tables.
 Consequently, we do not need to store for each~$S_x$ and each~$k'$ the value of~$Q_{S_x}[1,k']$.
 Instead, we recompute~$Q_{S_x}[1,k']$ if we need this value a second time. 
 Hence, the whole algorithm only uses polynomial space.
\end{proof}

As above, with a slight modification of the dynamic programming algorithm, we can improve on the running time for the special case of~\GLVC. 
Recall that~$I$ is a yes-instance of~\GLVC if and  only if, there is a \goodswap~$W$ for~$I$ with~$|W\cap S| \leq \frac{k+d}{2}$.
Hence for~\GLVC, it is sufficient in the computation of~$D^+_x[k']$, $D^-_x[k']$, and~$D^o_x[k']$ to only check for independent sets~$S_x$ in~$\beta(x)$ of size at most~$\min(k', \frac{k+d}{2})$.
Thus, we obtain the following.

\begin{corollary} 
\GLVC can be solved in $\Oh(\sw(G)^{\frac{k+d}{2}+1} \cdot k \cdot n + m)$~time, where $\sw(G)$~is the splitwidth of~$G$.
\end{corollary}

\section{Conclusion}
\label{sec:conclusion}
We introduced the notion of FPT running times that grow mildly with respect to some parameter~$\ell$ and strongly with respect to another parameter~$k$. Such running times are desirable in the setting where the parameter~$k$ is much smaller than~$\ell$. 
Parameterized local search is one scenario in which this assumption is certainly true, when~$k$ is the operational parameter that bounds the change between the current solution and its local neighbors. 
We showed that such running times are achievable for one of the most important graph problems in parameterized local search, \LVClong, and different well-known structural parameters taking the place of~$\ell$. 

There are numerous possibilities for future research.
First, it seems interesting to study further parameterized local search problems  with the aim of achieving FPT-algorithms whose running times grow strongly only with respect to the operational parameter~$k$. Such algorithms have been developed since the publication of the extended abstract of this work for local search versions of partitioning problems like \textsc{Max~$c$-Cut}~\cite{GGKM23,GMS25} and \textsc{Cluster Editing}~\cite{GMNW23,GMS25}. We also assume that some of the techniques presented in this work can be generalized to a local search version of~\textsc{0-1 Integer Linear Programming}.
More generally, algorithms with such running times could also be relevant in other scenarios with operational parameters~$k$, for example in turbo-charging of greedy algorithms~\cite{AKC+17,DSV22,GGJ+19}.

Second, it is open to improve the running time for the considered problems since our conditional lower bounds are not completely tight. For example, for~\LVClong parameterized by~$k$ and the~$h$-index, it is open whether a running time of~$\Oh(h^{k/2}\cdot n)$ is possible.
Third, it would be interesting to explore gap versions of further local search problems, both from a theoretical and a practical perspective.
In this context it would be interesting to explore whether there are parameters~$\ell$ for which the standard parameterized local search problem admits an FPT-algorithm but the gap version is~\W1-hard.  
Furthermore, it would be interesting to identify structural parameters~$\ell$ where permissive local search has an FPT-algorithm with running time~$\ell^{g(k)}\cdot n^{\Oh(1)}$ but non-permissive (often called \emph{strict}) local search does not.

Finally, it is open to further explore the concrete practical potential of our results: 
For \LVC, the FPT-algorithm with parameter~$(\Delta,k)$~\cite{KK17} was already quite efficient. Can it be improved using some of the techniques presented in this work? 
Similarly, can our theoretical results lead to good implementations of parameterized local search for \WVC? 
Here, an implementation of the FPT-algorithm for parameter~$(\Delta,k)$ (\Cref{running time rule small swaps weighted}) already gave very good results as a post-processing for state-of-the-art algorithms~\cite{Ull23}.
Would an implementation of, for example, the FPT-algorithm for~\WVC parameterized by~$h$-index and~$k$ (\Cref{thm: h index algo}) perform similarly well or even better?

\bibliographystyle{plain}
\bibliography{my_bib}

\end{document}